\documentclass[runningheads]{llncs}
\usepackage{graphicx}
\usepackage{amssymb,amsmath}

\usepackage{amsthm}
\usepackage[dvipsnames]{xcolor}
\usepackage[inline]{enumitem}
\usepackage{pifont}
\usepackage{listings}
\definecolor{codegreen}{rgb}{0,0.6,0.6}
\definecolor{codegray}{rgb}{0.5,0.5,0.5}
\definecolor{codepurple}{rgb}{0.58,0,0.82}
\definecolor{backcolour}{rgb}{0.95,0.95,0.92}

\usepackage{graphicx}
\usepackage{wrapfig}
\usepackage{tikz}
\usetikzlibrary{shapes.misc,arrows.meta}

\usepackage{float}
\usepackage{mathtools}
\usepackage{etoolbox}
\usepackage{comment}

\usepackage{environ}

\makeatletter
\newtheorem*{rep@theorem}{\rep@title}
\newcommand{\newreptheorem}[2]{%
\newenvironment{rep#1}[1]{%
 \def\rep@title{#2 \ref{##1}}%
 \begin{rep@theorem}}%
 {\end{rep@theorem}}}
\makeatother

\newreptheorem{theorem}{Theorem}
\newreptheorem{lemma}{Lemma}

\newif\ifcomments
\commentstrue

\ifcomments
\newcommand{\ucomment}[1]{\textcolor{red}{#1}}
\newcommand{\adcomment}[1]{\textcolor{orange}{AD: #1}}
\newcommand{\mpcomment}[1]{\textcolor{blue}{MP: #1}}
\newcommand{\mvcomment}[1]{\textcolor{green}{MV: #1}}
\newcommand{\pkcomment}[1]{\textcolor{CadetBlue}{PK: #1}}
\else
\newcommand{\ucomment}[1]{}
\newcommand{\adcomment}[1]{}
\newcommand{\mpcomment}[1]{}
\newcommand{\mvcomment}[1]{}
\newcommand{\pkcomment}[1]{}
\fi

\definecolor{drawioGreen}{HTML}{82B366}
\definecolor{drawioBlue}{HTML}{6C8EBF}
\definecolor{drawioRed}{HTML}{B85450}

\newcommand{\figlabel}[1]{\label{fig:#1}}
\newcommand{\figref}[1]{Fig.~\ref{fig:#1}}
\newcommand{\seclabel}[1]{\label{sec:#1}}
\newcommand{\secref}[1]{Section~\ref{sec:#1}}
\newcommand{\subseclabel}[1]{\label{subsec:#1}}
\newcommand{\subsecref}[1]{Section~\ref{subsec:#1}}
\newcommand{\exlabel}[1]{\label{ex:#1}}
\newcommand{\exref}[1]{Example~\ref{ex:#1}}

\newcommand{\thmlabel}[1]{\label{thm:#1}}
\newcommand{\thmref}[1]{Theorem~\ref{thm:#1}}
\newcommand{\proplabel}[1]{\label{prop:#1}}
\newcommand{\propref}[1]{Proposition~\ref{prop:#1}}
\newcommand{\lemlabel}[1]{\label{lem:#1}}
\newcommand{\lemref}[1]{Lemma~\ref{lem:#1}}

\newcommand{\applabel}[1]{\label{app:#1}}
\newcommand{\appref}[1]{Appendix~\ref{app:#1}}

\newcommand{\nats}{\mathbb{N}}

\newcommand{\Aa}{\mathcal{A}}
\newcommand{\Bb}{\mathcal{B}}
\newcommand{\Cc}{\mathcal{C}}

\newcommand{\Ff}{\mathcal{F}}
\newcommand{\Gg}{\mathcal{G}}

\newcommand{\Ii}{\mathcal{I}}

\newcommand{\Mm}{\mathcal{M}}
\newcommand{\Pp}{\mathcal{P}}

\newcommand{\Rr}{\mathcal{R}}

\newcommand{\Ss}{\mathcal{S}}

\renewcommand{\vec}[1]{\overline{#1}}
\newcommand{\set}[1]{\{#1\}}
\DeclarePairedDelimiterX\bigset[1]\lbrace\rbrace{\def\given{\;\delimsize\vert\;}#1}
\newcommand{\setpred}[2]{\{#1 \,\mid\, #2\}}

\newcommand{\angular}[1]{\langle #1 \rangle}
\newcommand{\tuple}{\angular}
\newcommand{\delequal}{\overset{\Delta}{=}}
\renewcommand{\emptyset}{\varnothing}

\newcommand{\xdownarrow}[1]{%
  {\left\downarrow\vbox to #1{}\right.\kern-\nulldelimiterspace}
}

\newcommand{\production}[1]{\angular{#1}}
\newcommand{\produces}{\to}
\newcommand{\stmt}{\angular{stmt}}
\newcommand{\cond}{\angular{cond}}

\newcommand{\code}[1]{\texttt{#1}}
\newcommand{\codekey}[1]{\textbf{#1}}
\newcommand{\cd}[1]{\code{#1}}
\newcommand{\cdm}[1]{\mathtt{#1}}
\newcommand{\pskip}{\codekey{skip}}
\newcommand{\passume}{\codekey{assume}}
\newcommand{\passert}{\codekey{assert}}
\newcommand{\pif}{\codekey{if}}
\newcommand{\pthen}{\codekey{then}}
\newcommand{\pelse}{\codekey{else}}
\newcommand{\pwhile}{\codekey{while}}
\newcommand{\passign}{:=}
\newcommand{\poutputs}{\Rightarrow}

\newcommand{\pfalse}{\ensuremath{\bot}}
\newcommand{\pite}{\codekey{ite}}
\newcommand{\pseq}{\codekey{seq}}
\newcommand{\proot}{\codekey{root}}
\newcommand{\pcheck}{\codekey{check}}

\newcommand{\tsexec}{\textsf{Exec}}
\newcommand{\tscexec}{\textsf{CompExec}}
\newcommand{\tsfalse}{\bot}
\newcommand{\prog}{\textsf{Prog}}
\newcommand{\trlbl}{\textsf{Tree}}
\newcommand{\exec}{\textsf{Exec}}
\newcommand{\pexec}{\textsf{PExec}}
\newcommand{\comp}{\textsf{T}}
\newcommand{\Terms}{\textsf{Terms}}

\newcommand{\init}[1]{\widehat{#1}}

\newcommand{\dblqt}[1]{\text{``}#1\text{''}}

\newcommand{\fixperm}[1]{\angular{#1}}
\newcommand{\out}[1]{\textbf{\textsf{o}}_{#1}}
\newcommand{\pgm}{\angular{pgm}}
\newcommand{\pcall}{\codekey{call}}
\newcommand{\preturn}{\codekey{return}}
\newcommand{\goesto}{\rightarrow}

\newcommand{\reject}{q_\mathsf{reject}}

\newcommand{\exptime}{\mathsf{EXPTIME}}
\newcommand{\twoexptime}{2\mathsf{EXPTIME}}
\newcommand{\aexpspace}{\mathsf{AEXPSPACE}}

\newcommand{\sketch}{\textsc{Sketch}}
\newcommand{\sygus}{\textsc{SyGuS}}

\newcommand{\gram}{\Gg}
\newcommand{\cohcor}{\textsf{cc}}
\newcommand{\cohcortopdown}{\textsf{cc-td}}
\newcommand{\cohcorexec}{\textsf{exec}}

\newcommand*{\tree}{\tikz{
    \node[shape=circle, fill=black, inner sep=0pt] (a) at (0,0) {} ;
    \node[shape=circle, fill=black, inner sep=0pt] (b) at (-0.05,-0.07) {} ;
    \node[shape=circle, fill=black, inner sep=0pt] (c) at (0.05,-0.07) {} ;
    \node[shape=circle, fill=black, inner sep=0pt] (d) at (-0.08,-0.14) {} ;
    \node[shape=circle, fill=black, inner sep=0pt] (e) at (-0.02,-0.14) {} ;
    \draw[-, line width=0.05mm] (a) to (b);
    \draw[-, line width=0.05mm] (a) to (c);
    \draw[-, line width=0.05mm] (b) to (d);
    \draw[-, line width=0.05mm] (b) to (e);
}}

\newcommand*{\word}{\tikz{
    \node[shape=circle, fill=black, inner sep=0pt] (a) at (0,0) {} ;
    \node[shape=circle, fill=black, inner sep=0pt] (b) at (0.08, 0) {} ;
    \node[shape=circle, fill=black, inner sep=0pt] (c) at (0.16,0) {} ;
    \node[shape=circle, fill=white, inner sep=0pt] (d) at (0.12,-0.1) {} ; 
    \draw[-, line width=0.05mm] (a) to (b);
    \draw[-, line width=0.05mm] (b) to (c);
}}

\newcommand{\treeaut}[1]{\mathcal{#1}^{\tree}}
\newcommand{\wordaut}[1]{\mathcal{#1}^{\word}}
\newcommand{\lab}{\gamma}

\newcommand{\ldir}{L}
\newcommand{\rdir}{R}
\newcommand{\udir}{U}

\newcommand{\ddir}{D}
\newcommand{\uldir}{U_L}
\newcommand{\urdir}{U_R}

\newcommand{\schema}{\Ss}
\newcommand{\schemant}{PN}
\newcommand{\schemapv}{PV}

\newcommand{\slp}{\textsf{SLP}}

\newcommand{\superterm}[3]{#1 \preccurlyeq_{#2} #3}

\newcommand{\strategy}{\sigma}
\newcommand{\str}{\strategy}
\newcommand{\strategies}{\mathfrak{S}}
\newcommand{\stannot}{\mathfrak{a}}
\newcommand{\stconannot}{\overline{\stannot}}
\newcommand{\bag}{\mathit{Bag}}

\newcommand{\labts}{\lambda}

\newcommand*\circled[1]{\tikz[baseline=(char.base)]{
            \node[shape=circle,draw,fill=blue!20!white,inner sep=0.5pt] (char) {\tiny \texttt{#1} };}}
\newcommand{\holesymb}{\texttt{??}}
\newcommand{\holebr}[1]{\angular{\angular{#1}}}
\newcommand{\holetext}[1]{\holesymb^{#1}}
\newcommand{\hole}[1]{\holebr{\,\holetext{#1}\,}}
\newcommand{\holepred}[2]{\holebr{\,\holetext{#1}\,|\,\text{#2}\,}}
\renewcommand*\boxed[2]{\tikz[baseline=(char.base)]{
            \node[shape=rectangle, draw, dashed, inner sep=2pt, minimum height=10pt, minimum width = #1pt, align=left] (char) { #2 };}
            }

\newtoggle{linear}
\toggletrue{linear}

\iftoggle{linear}{%
  \newcommand{\complexitygrammar}{linear}
  \newcommand{\ordernotation}{O(|\treeaut{\Aa}|) = O(2^{2^{\text{poly}(|V|)}}\cdot|\gram|)}
}{%
  \newcommand{\complexitygrammar}{polynomial}
  \newcommand{\ordernotation}{O(\text{poly}(|\treeaut{\Aa}|)) = O(2^{2^{\text{poly}(|V|)}}\cdot|\gram|)}
}

\newcounter{exec}
\newcommand{\ctr}{\arabic{exec}}
\newcommand{\incctr}{\addtocounter{exec}{1}}

\begin{document}
\title{Decidable Synthesis of Programs with Uninterpreted Functions
\thanks{
Paul Krogmeier and Mahesh Viswanathan are partially supported by NSF CCF 1901069.
Umang Mathur is partially supported by a Google PhD Fellowship.
}
}
%
%
\author{
Paul Krogmeier \orcidID{0000-0002-6710-9516}
\and
Umang Mathur \orcidID{0000-0002-7610-0660}
\and
Adithya Murali \orcidID{0000-0002-6311-1467}
\and
P. Madhusudan
\and
Mahesh Viswanathan
}
\authorrunning{P. Krogmeier et al.}
%
\institute{University of Illinois at Urbana-Champaign, USA\\
\email{\{paulmk2,umathur3,adithya5,madhu,vmahesh\}@illinois.edu}}

\maketitle              

\begin{abstract}





We identify a decidable synthesis problem for a class of programs of
unbounded size with conditionals and iteration that work over infinite
data domains. The programs in our class use uninterpreted functions
and relations, and abide by a restriction called coherence that was
recently identified to yield decidable verification.  We formulate a
powerful grammar-restricted (syntax-guided) synthesis problem for
coherent uninterpreted programs, and we show the problem to be
decidable, identify its precise complexity, and also study several
variants of the problem.

\end{abstract}


\section{Introduction}
\seclabel{intro}

Program synthesis is a thriving area of research that addresses the
problem of automatically constructing a program that meets a
user-given
specification~\cite{ProgramSynthGPS,ProgramSynthInductiveProg,sygus}.
Synthesis specifications can be expressed in various ways: as
input-output examples~\cite{flashfill11,flashfill12}, temporal logic
specifications for reactive programs~\cite{PR89}, logical
specifications~\cite{sygus,SearchBasedProgramSynthesisCACM}, etc.
Many targets for program synthesis exist, ranging from transition
systems~\cite{PR89,KMTV00}, logical expressions~\cite{sygus},
imperative programs~\cite{sketching}, distributed transition
systems/programs~\cite{PR90,madhudistsynth,igoranca14}, filling holes
in programs~\cite{sketching}, or repairs of
programs~\cite{singhrepair}.

A classical stream of program synthesis research is one that emerged
from a problem proposed by Church~\cite{church60} in 1960 for Boolean
circuits. Seminal results by B\"uchi and
Landweber~\cite{BuchiLandweber69} and Rabin~\cite{Rabin72} led to a
mature understanding of the problem, including connections to infinite
games played on finite graphs and automata over infinite trees
(see~\cite{automata-logics-games,kpvPneuli}).  Tractable synthesis for temporal
logics like LTL, CTL, and their fragments was investigated and several
applications for synthesizing hardware circuits
emerged~\cite{BJPPS12,BGJPPW07}.

In recent years, the field has taken a different turn, tackling
synthesis of programs that work over infinite domains such as
strings~\cite{flashfill11,flashfill12},
integers~\cite{sketching,sygus}, and heaps~\cite{QS17}.  Typical
solutions derived in this line of research involve (a) bounding the
class of programs to a finite set (perhaps iteratively increasing the
class) and (b) searching the space of programs using techniques like
symmetry-reduced enumeration, SAT solvers, or even random
walks~\cite{SearchBasedProgramSynthesisCACM,sygus}, typically guided
by counterexamples (CEGIS)~\cite{sketching,LMN16,JhaSeshia17}.  Note
that iteratively searching larger classes of programs allows synthesis
engines to find a program if one exists, but it does not allow one to
conclude that there is no program that satisfies the specification.
Consequently, in this stream of research, decidability results are uncommon (see~\secref{related} for some exceptions in certain heavily
restricted cases).

\emph{In this paper we present, to the best of our knowledge, the
  first decidability results for program synthesis over a natural
  class of programs with iteration/recursion, having arbitrary sizes,
  and which work on infinite data domains. In particular, we show
  decidable synthesis of a subclass of programs that use uninterpreted
  functions and relations.}

Our primary contribution is a decidability result for realizability
and synthesis of a restricted class of imperative \emph{uninterpreted}
programs.  Uninterpreted programs work over infinite data models that
give arbitrary meanings to their functions and relations. Such
programs satisfy their assertions if they hold along all executions
for \emph{every} model that interprets the functions and relations.
The theory of uninterpreted functions and relations is well studied---classically, in 1929, by G\"odel, where completeness results were
shown~\cite{godelcompleteness} and, more recently, its decidable
quantifier-free fragment has been exploited in SMT solvers in
combination with other theories~\cite{calcofcomputation}.  In recent
work~\cite{MMV19}, a subclass of uninterpreted programs, called
\emph{coherent} programs, was identified and shown to have a decidable
verification problem.  Note that in this verification problem there
are no user-given loop invariants; the verification algorithm finds
inductive invariants and proves them automatically in order to prove
program correctness.

In this paper, we consider the synthesis problem for coherent
uninterpreted programs.  The user gives a \emph{grammar} $\gram$ that
generates well-formed programs in our programming language.  The
grammar can force programs to have $\passert$ statements at various
points which collectively act as the specification. The program
synthesis problem is then to construct a coherent program, if one
exists, conforming to the grammar $\gram$ that satisfies all
assertions in all executions when running on \emph{any} data model
that gives meaning to function and relation symbols.

Our primary result is that the realizability problem (checking the
existence of a program conforming to the grammar and satisfying its
assertions) is decidable for coherent uninterpreted programs.  We
prove that the problem is $\twoexptime$-complete.
Further, whenever a
correct coherent program that conforms to the grammar exists, we can
synthesize one.  We also show that the realizability/synthesis problem
is undecidable if the coherence restriction is dropped. In fact we
show a stronger result that the problem is undecidable even for
synthesis of \emph{straight-line} programs (without conditionals and
iteration)!

Coherence of programs is a technical restriction that was introduced
in~\cite{MMV19}. It consists of two properties, both of which were
individually proven to be essential for ensuring that program
verification is decidable. Intuitively, the restriction demands that
functions are computed on any tuple of terms only once and that
assumptions of equality come early in the executions.
In more recent work~\cite{mathur2019forest}, the authors extend this
decidability result to handle map updates, and applied it to memory
safety verification for a class of heap-manipulating programs on
forest data-structures, demonstrating that the restriction of
coherence is met in practice by certain natural and useful classes of
programs.

Note that automatic synthesis of correct programs over infinite
domains demands that we, at the very least, can automatically verify
the synthesized program to be correct. The class of coherent
uninterpreted programs identified in the work of~\cite{MMV19} is the
only natural class of programs we are aware of that has recursion and
conditionals, works over infinite domains, and admits decidable
verification.  Consequently, this class is a natural target for
proving a decidable synthesis result.

The problem of synthesizing a program from a grammar with assertions
is a powerful formulation of program synthesis.  In particular, the
grammar can be used to restrict the space of programs in various ways.
For example, we can restrict the space syntactically by disallowing
while loops.  Or, for a fixed $n$, by using a set of Boolean variables
linear in $n$ and requiring a loop body to strictly increment a
counter encoded using these variables, we can demand that loops
terminate in a linear/polynomial/exponential number of iterations.  We
can also implement loops that do not always terminate, but terminate
only when the data model satisfies a particular property, e.g.,
programs that terminate only on finite list segments, by using a
skeleton of the form: $\pwhile~(\cd{x} \neq
\cd{y})$\{$~ ...~ ;~ $\cd{x}~\passign~\cd{next(x)}$\}$.
Grammar-restricted program synthesis can express the synthesis of
programs with holes, used in systems like {\sketch}~\cite{sketch},
where the problem is to fill holes using programs/expressions
conforming to a particular grammar so that the assertions in the
program hold.
Synthesizing programs or expressions using restricted grammars is also
the cornerstone of the intensively studied {\sygus} (syntax-guided
synthesis) format~\cite{sygus,syguswebsite}~\footnote{Note, however,
  that both {\sketch} and {\sygus} problems are defined using
  functions and relations that are interpreted using standard theories
  like arithmetic, etc., and hence of course do not have decidable
  synthesis.}.

The proof of our decidability result relies on tree automata, a
callback to classical theoretical approaches to synthesis.  The key
idea is to represent programs as trees and build automata that accept
trees corresponding to correct programs.  The central construction is
to build a two-way alternating tree automaton that accepts \emph{all}
program trees of coherent programs that satisfy their assertions.
Given a grammar $\gram$ of programs (which has to satisfy certain
natural conditions), we show that there is a regular set of program
trees for the language of allowed programs $L(\gram)$.  Intersecting
the automata for these two regular tree languages and checking for
emptiness establishes the upper bound.  Our constructions crucially
use the automaton for verifying coherent uninterpreted programs
in~\cite{MMV19} and adapt ideas from~\cite{csl11} for building two-way
automata over program trees.  
Our final decision procedure is doubly-exponential in the number of
program variables and \emph{\complexitygrammar} in the size of the
grammar. We also prove a matching lower bound by reduction from the
acceptance problem for alternating exponential-space Turing machines.
%
%
The reduction is non-trivial in that programs (which correspond to
runs in the Turing machine) must simulate sequences of configurations,
each of which is of exponential size, by using only polynomially-many
variables.

\subsubsection*{Recursive Programs, Transition Systems, and Boolean programs:}~\\
We study three related synthesis problems. First, we show that our
results extend to synthesis of call-by-value \emph{recursive}
uninterpreted programs (with a fixed number of functions and fixed
number of local/global variables).  This problem is also
$\twoexptime$-complete but is more complex, as even single executions
simulated on the program tree must be split into separate copies, with
one copy executing the summary of a function call and the other
proceeding under the assumption that the call has returned in a
summarized state.

We next examine a synthesis problem for \emph{transition
  systems}. Transition systems are similar to programs in that they
execute similar kinds of atomic statements. We allow the user to
restrict the set of allowable executions (using regular sets). Despite
the fact that this problem seems very similar to program synthesis, we
show that it is an \emph{easier} problem, and coherent transition
system realizability and synthesis can be solved in time exponential
in the number of program variables and polynomial in the size of the automata that
restrict executions. We prove a corresponding lower bound to establish
$\exptime$-completeness of this problem.

Finally, we note that our results also show, as a corollary, that the
grammar-restricted realizability/synthesis problem for Boolean
programs (resp. execution-restricted synthesis problem for Boolean
transition systems)
is decidable and is $\twoexptime$-complete (resp.
$\exptime$-complete). These results for Boolean programs are
themselves new. The lower bound results for these problems hence show
that coherent program/transition-system synthesis is not particularly
harder than Boolean program synthesis for uninterpreted
programs. Grammar-restricted Boolean program synthesis is an important
problem which is addressed by many practical synthesis systems like
Sketch~\cite{sketch}.

\medskip Due to space restrictions, we present only proof gists for
main results in the paper. All the complete proofs
 can be found in our technical report~\cite{techreport}.




\section{Examples}
\seclabel{examples}


We will begin by looking at several examples to gain some intuition
for uninterpreted programs.
\begin{figure}[h]
\begin{minipage}{0.5\textwidth}
\begin{align*}
\begin{array}{l}
\cd{cipher} \passign \cd{enc(secret, key)}; \\
\passume (\cd{secret} = \cd{dec(cipher, key)}); \\
\holepred{}{Cannot refer to \cd{secret} or \cd{key}} ; \\
\passert (\cd{z} = \cd{secret})
\end{array}
\end{align*} \\
\indent
\quad \quad Decrypting a ciphertext
\end{minipage}
\vline
\begin{minipage}{0.5\textwidth}
\begin{align*}
\begin{array}{l}
    \passume (\cd{T} \neq \cd{F}); \\
    \pif~\code{(x = T)}~\pthen~\cd{b} \passign \cd{T}~\pelse~\cd{b} \passign \cd{F};\\
    \holepred{}{Cannot refer to \code{x} or \code{b}}; \\
    \passert\code{(y = b)}
\end{array}
\end{align*} \\
\indent
\quad Synthesis with incomplete information
\end{minipage}
\caption{Examples of programs with holes}
\figlabel{motivating-examples}
\end{figure}

\begin{example}
\exlabel{crypto-ex}
Consider the program in~\figref{motivating-examples} (left).
This program has a \emph{hole} `$\holepred{}{Cannot \ldots }$' that we
intend to fill with a sub-program so that the entire program (together
with the contents of the hole) satisfies the assertion at the end.
The sub-program corresponding to the hole is allowed to use the
variable \cd{cipher} as well as some additional variables
$\cd{y}_1, \ldots, \cd{y}_n$ (for some fixed $n$), but is not allowed
to refer to $\cd{key}$ or $\cd{secret}$ in any way. Here we also
restrict the hole to exclude while loops. This example models the
encryption of a secret message $\cd{secret}$ with a key $\cd{key}$.
The assumption in the second line of the program models the fact that
the secret message can be decrypted from $\cd{cipher}$ and $\cd{key}$.
Here, the functions $\cd{enc}$ and $\cd{dec}$ are \emph{uninterpreted
  functions}, and thus the program we are looking for is an
\emph{uninterpreted program}.  For such a program, the assertion
$\dblqt{\passert (\cd{z} = \cd{secret})}$ holds at the end if it holds
for \emph{all models}, i.e, for all interpretations of \cd{enc} and
\cd{dec} and for all initial values of the program variables
$\cd{secret}$, $\cd{key}$, $\cd{cipher}$, and
$\cd{y}_1, \ldots, \cd{y}_n$.  With this setup, we are essentially
asking whether a program that does not have access to $\cd{key}$ can
recover $\cd{secret}$.  It is not hard to see that there is no program
which satisfies the above requirement.  The above modeling of keys,
encryption, nonces, etc. is common in algebraic approaches to modeling
cryptographic
protocols~\cite{dolev1983security,Durgin2004}. 


\end{example}

\begin{example}
  \exlabel{incomplete-info-ex} The program
  in~\figref{motivating-examples} (right) is another simple example of
  an unrealizable specification.  
  The program variables here are $\cd{x}, \cd{b}$, and $\cd{y}$.
  The hole in this partial program is restricted so that it cannot
  refer to $\cd{x}$ or $\cd{b}$.  It is easy to phrase the question
  for synthesis of the complete program in terms of a grammar.  The
  restriction on the hole ensures that the synthesized code fragment
  can neither directly check if $\cd{x} = \cd{T}$, nor indirectly
  check via $\cd{b}$.  Consequently, it is easy to see that there is
  no program for the hole that can ensure $\cd{y}$ is equal to
  $\cd{b}$.  We remark that the code at the hole, apart from not being
  allowed to examine some variables, is also implicitly prohibited
  from looking at the control path taken to reach the hole.  If we
  could synthesize two different programs depending on the control
  path taken to reach the hole, then we could set
  $\cd{y} \passign \cd{T}$ when the \pthen-branch is taken and set
  $\cd{y} \passign \cd{F}$ when the \pelse-branch is taken.  Program
  synthesis requires a control-flow independent decision to be made
  about how to fill the hole. In this sense, we can think of the hole
  as having only \emph{incomplete information} about the executions
  for which it must be correct. This can be used to encode
  specifications using complex ghost code, as we show in the next
  examples.  In~\secref{further-results}, we explore a slightly
  different synthesis problem, called \emph{transition system
    synthesis}, where holes can be differently instantiated based on
  the history of an execution.
\end{example}

\begin{example}
  \exlabel{list-search} In this example, we model the synthesis of a
  program that checks whether a linked list pointed to by some node
  \cd{x} has a key \cd{k}.  We model a $\textit{next}$ pointer with a
  unary function \cd{next} and we model locations using elements in
  the underlying data domain.

  Our formalism allows only for $\passert$ statements to specify
  desired program properties.
  In order to state the correctness specification for our desired list-search program,
  we interleave \emph{ghost code}
  into the program skeleton;
  we distinguish ghost code fragments by enclosing them in \boxed{5}{dashed boxes}.
  The skeleton in \figref{ghost-code} has a
  loop that advances the pointer variable \cd{x} along the list until
  \cd{NIL} is reached. We model \cd{NIL} with an immutable program
  variable. The first hole `$\hole{\circled{1}}$' before the
  \pwhile-loop and the second hole `$\hole{\circled{2}}$' within the
  \pwhile-loop need to be filled so that the assertion at the end is
  satisfied. We use three ghost variables in the skeleton:
  $\cdm{g_{ans}}$, $\cdm{g_{witness}}$, and $\cdm{g_{found}}$. The
  ghost variable $\cdm{g_{ans}}$ evaluates to whether we expect to
  find $\cd{k}$ in the list or not, and hence at the end the skeleton
  asserts that the Boolean variable $\cd{b}$ computed by the holes is
  precisely $\cdm{g_{ans}}$.  The holes are restricted to
  not look at the ghost variables.
\begin{flushleft}
\begin{minipage}[t]{0.57\linewidth}
  Now, notice that the skeleton needs to \emph{check} that the answer
  $\cdm{g_{ans}}$ is indeed correct.  If $\cdm{g_{ans}}$ is not
  $\cd{T}$, then we add the assumption that $\cd{key(x)} \neq \cd{k}$
  in each iteration of the loop, hence ensuring the key is not
  present.  For ensuring correctness in the case
  $\cdm{g_{ans}}= \cd{T}$, we need two more ghost variables
  $\cdm{g_{witness}}$ and $\cdm{g_{found}}$.  The variable
  $\cdm{g_{witness}}$ witnesses the precise location in the list that
  holds the key $\cd{k}$, and variable $\cdm{g_{found}}$ indicates
  whether the location at $\cdm{g_{witness}}$ belongs to the list
  pointed to by $\cd{x}$.  Observe that this specification can be
  realized by filling `$\hole{\circled{1}}$' with
  ``$\cd{b} \passign \cd{F}$'' and `$\hole{\circled{2}}$' with
  ``$\pif~ \code{key(x)} = \cd{k} ~\pthen~ \cd{b} \passign \cd{T}$'',
  for instance.  Furthermore, this program is
  \emph{coherent}~\cite{MMV19} and hence our decision procedure will
  answer in the affirmative and synthesize code for the holes.
\end{minipage} \hfill
\begin{minipage}[t]{0.4\linewidth}
\begin{flushright}
    \centering
    \vspace{-0.5in}
    \begin{figure}[H]
    \begin{align*}
      \begin{array}{l}
        \passume (\cd{T} \neq \cd{F}) ; \\
        \vspace{0.1cm}
        \boxed{40}{$\cdm{g_{found}} \passign \code{F};$}\\
        \vspace{0.1cm}
        \hole{\circled{1}}; \\
        \pwhile \cdm{(x \neq NIL)} \,\{\\
        \vspace{0.1cm}
        \hspace*{0.4cm}
        \boxed{120}{
        $\pif~ (\cdm{g_{ans}} \neq \cd{T})~\pthen$ \\
        \hspace*{0.4cm}~$\passume (\cd{key(x)} \neq \cd{k});$ \\
        $\pelse~\pif~(\cdm{g_{witness}} = \cd{x})~\pthen~ \{$\\
        \hspace*{0.4cm}~$\passume~(\cd{key(x) = k});$ \\
        \hspace*{0.4cm}~$\cdm{g_{found}} \passign \cd{T};$ \\
        \hspace*{0.1cm} $\} ;$
        }\\
        \vspace{0.1cm}
        \hspace*{0.4cm}~\hole{\circled{2}} ;\\
        \hspace*{0.4cm}~\cd{x} \passign~ \cd{next(x)} ;\\
        \hspace*{0.1cm}\}\\
        \vspace{0.1cm}
        \boxed{80}{$\passume~(\cdm{g_{ans} = T}  \Rightarrow \cdm{g_{found} = T});$}\\
        \passert~ \code{b} = \cd{T} \iff \cdm{g_{ans}} = \cd{T}
      \end{array}
    \end{align*}
    \vspace{-0.2in}
    \caption{Skeleton with ghost code}
    \label{fig:ghost-code}
  \end{figure}
\end{flushright}
\end{minipage}
\end{flushleft}
\vspace{-0.2in} In fact, our procedure will synthesize a
representation for \emph{all} possible ways to fill the holes (thus
including the solution above) and it is therefore possible to
enumerate and pick specific solutions.  It is straightforward to
formulate a grammar which matches this setup.  As noted, we must
stipulate that the holes do not use the ghost variables.
\end{example}

\vspace{-0.1in}
\begin{example}
  \exlabel{list-search-unrealizable} Consider the same program
  skeleton as in~\exref{list-search}, but let us add an assertion at
  the end:
  ``$\passert~(\cd{b = T} \Rightarrow \cd{z}=\cdm{g_{witness}})$'',
  where $\cd{z}$ is another program variable.  We are now demanding
  that the synthesized code also find a location $\cd{z}$, whose key
  is $\cd{k}$, that is equal to the ghost location
  $\cdm{g_{witness}}$, which is guessed nondeterministically at the
  beginning of the program. This specification is \emph{unrealizable}:
  for a list with multiple locations having the key $\cd{k}$, no
  matter what the program picks we can always take $\cdm{g_{witness}}$
  to be the \emph{other} location with key $\cd{k}$ in the list, thus
  violating the assertion.  Our decision procedure will report in the
  negative for this specification.
\end{example}

\vspace{0.1in}
\begin{example}[Input/Output Examples]
  We can encode input/output examples by adding a sequence of
  assignments and assumptions that define certain models at the
  beginning of the program grammar.
  For instance, the sequence of statements
  in \figref{input-output}
    defines a linked list of two elements with different keys.
  \vspace{-0.3in}
  \begin{flushleft}
  \begin{minipage}{0.6\linewidth}
    We can similarly use special variables to define the output that
    we expect in the case of each model. And as we saw in the ghost
    code of \figref{ghost-code}, we can use fresh variables to
    introduce nondeterministic choices, which the grammar can use to
    pick an example model nondeterministically.  Thus when the
    synthesized program is executed on the chosen model it computes
    the expected answer.  This has the effect of requiring a solution
    that generalizes across models. See~\cite{techreport} for a more
    detailed example. 
  \end{minipage} \hfill
  \begin{minipage}{0.38\linewidth}
    \begin{flushright}
      \begin{figure}[H]
        \begin{align*}
          \begin{array}{l}
            \passume(\cdm{x_1} \neq \cd{NIL});~ \\
            \cdm{x_2} \passign \cdm{next(x_1)};~ \\
            \passume(\cdm{x_2} \neq \cd{NIL});~ \\
            \passume(\cdm{next(x_2)} = \cd{NIL});~ \\
            \cdm{k_1} \passign \cdm{key(x_1)};~ \\
            \cdm{k_2} \passign \cdm{key(x_2)};~ \\
            \passume(\cdm{k_1} \neq \cdm{k_2})
          \end{array}
        \end{align*}
        \caption{An example model}
        \label{fig:input-output}
      \end{figure}
    \end{flushright}
  \end{minipage}
\end{flushleft}
\end{example}


\vspace{-0.2in}
\section{Preliminaries}
\seclabel{prelim}

In this section we define the syntax and semantics of uninterpreted
programs and the \emph{(grammar-restricted) uninterpreted program
  synthesis} problem.

\subsubsection{Syntax}
\seclabel{syntax}

We fix a first order signature $\Sigma = (\Ff, \Rr)$, where $\Ff$ and
$\Rr$ are sets of function and relation symbols, respectively.  Let
$V$ be a finite set of program variables.  The set of programs over
$V$ is inductively defined using the following grammar, with
$f \in \Ff$, $R \in \Rr$ (with $f$ and $R$ of the appropriate
arities), and $x, y, z_1, \ldots, z_r \in V$.
\begin{align*}
\stmt_V ::=\,\,&
\,\,  \pskip \,
\mid \, x \passign y \,
\mid \, x \passign f(z_1, \ldots, z_r) \,
\mid \, \\ &\passume \, \big(\cond_V\big) \,
\mid \, \passert \, \big(\cond_V\big)
\mid \, \stmt_V \, ;\, \stmt_V
\mid \, \\ &\pif \, \big(\cond_V\big) \, \pthen \, \stmt_V \, \pelse \, \stmt_V \,
\mid \, \pwhile \, \big(\cond_V\big) \, \stmt_V
\\
\cond_V ::=\,\,&
\, x = y \,
\mid \, R(z_1, \ldots,  z_r) \,
\mid \, \cond_V \lor \cond_V \,
\mid \, \neg \cond_V
\end{align*}
Without loss of generality, we can assume that our programs do not use
relations (they can be modeled with functions) and that every
condition is either an equality or disequality between variables
(arbitrary Boolean combinations can be modeled with nested
$\pif{-}\pthen{-}\pelse$). When the set of variables $V$ is clear from
context, we will omit the subscript $V$ from $\stmt_V$ and $\cond_V$.

\subsubsection{Program Executions}
An execution over $V$ is a finite word over the alphabet
\begin{align*}
  \Pi_V = \setpred{
  \dblqt{x \passign y},
  \dblqt{x \passign f(\vec{z})},
  &\dblqt{\passume(x = y)},
  \dblqt{\passume(x \neq y)},
  \\ &\dblqt{\passert(\pfalse)}
  }{
  x, y\in V, \vec{z}\in V^r, f \in \Ff
  }.
\end{align*}

The set of \emph{complete executions} for a program $p$ over $V$,
denoted $\exec(p)$, is a regular language. See~\cite{techreport} for a
straightforward definition. 
%
The set $\pexec(p)$ of \emph{partial executions} is the set of
prefixes of complete executions in $\exec(p)$. We refer to partial
executions as simply \emph{executions}, and clarify as needed when the
distinction is important.

\vspace{-0.15in}
\subsubsection{Semantics}
The semantics of executions is given in terms of data models.  A data
model $\Mm = (U, \Ii)$ is a first order structure over $\Sigma$
comprised of a universe $U$ and an interpretation function $\Ii$ for
the program symbols. The semantics of an execution $\pi$ over a data
model $\Mm$ is given by a configuration $\sigma(\pi, \Mm) : V \to U$
which maps each variable to its value in the universe $U$ at the end
of $\pi$. This notion is straightforward and we skip the formal
definition (see~\cite{MMV19} for details). For a fixed program $p$,
any particular data model corresponds to at most one complete
execution $\pi\in\exec(p)$.

An execution $\pi$ is \emph{feasible} in a data model $\Mm$ if for
every prefix $\rho = \rho' \cdot \passume(x \sim y)$ of $\pi$ (where
$\sim \,\in \set{=,\neq}$), we have
$\sigma(\rho', \Mm)(x) \sim \sigma(\rho', \Mm)(y)$.  Execution $\pi$
is said to be \emph{correct} in a data model $\Mm$ if for every prefix
of $\pi$ of the form $\rho = \rho' \cdot \passert(\pfalse)$, we have
that $\rho'$ is not feasible, or \emph{infeasible} in $\Mm$.  Finally,
a program $p$ is said to be \emph{correct} if for all data models
$\Mm$ and executions $\pi \in \pexec(p)$, $\pi$ is correct in $\Mm$.

\subsection{The Program Synthesis Problem}
\seclabel{problem-statement} We are now ready to define the program
synthesis problem. Our approach will be to allow users to specify a
grammar and ask for the synthesis of a program from the grammar. We
allow the user to express specifications using \emph{assertions} in
the program to be synthesized.

\vspace{-0.15in}
\subsubsection{Grammar Schema and Input Grammar.}
\seclabel{schema} In our problem formulation, we allow users to define
a grammar which conforms to a schema, given below. The input grammars
allow the usual context-free power required to describe proper
nesting/bracketing of program expressions, but disallow other uses of
the context-free power, such as \emph{counting statements}.

\begin{flushleft}
\begin{minipage}[t]{0.6\linewidth}
  For example, we disallow the grammar in
  \figref{disallowed-grammar}. This grammar has two non-terminals $S$
  (the start symbol) and $T$. It generates programs with a conditional
  that has the \emph{same} number of assignments in the $\pif$ and
  $\pelse$ branches.  We assume a countably infinite set $\schemant$
  of nonterminals and a countably infinite set $\schemapv$ of program
  variables.  The grammar schema $\schema$ over $\schemant$ and
  $\schemapv$ is an infinite collection of productions:
\end{minipage} \hfill
\begin{minipage}[t]{0.37\linewidth}
  \vspace{-0.3in}
  \begin{flushright}
    \begin{figure}[H]
      \begin{alignat*}{2}
        S &\produces\,\, &&\pif~(\cd{x} = \cd{y}) \\
        & &&\pthen~\cd{u} \passign \cd{v}~T~\cd{u} \passign \cd{v} \\
        T &\produces\,\, &&\pelse  \\
        T &\produces\,\, &&;~\cd{u} \passign \cd{v}~~T~~\cd{u} \passign
        \cd{v}~;~
      \end{alignat*}
      \vspace{-0.2in}
      \caption{Grammar with counting}
      \label{fig:disallowed-grammar}
    \end{figure}
  \end{flushright}
\end{minipage}
\end{flushleft}


\vspace{-0.3in}
\begin{align*}
\schema = \bigset*{
\begin{aligned}
\begin{array}{l}
\dblqt{P \produces \, x \passign y},
\dblqt{P \produces \, x \passign f(\vec{z})}, \\
\dblqt{P \produces \, \passume(x \sim y)},
\dblqt{P \produces \, \passert(\pfalse)}, \\
\dblqt{P \produces \, \pskip},
\dblqt{P \produces \, \pwhile \, (x \sim y) \,\, P_1}, \\
\dblqt{P \produces \, \pif \, (x \sim y) \, \pthen \, P_1 \, \pelse \, P_2},
\dblqt{P \produces \, P_1; P_2}
\end{array}
\end{aligned}
\given
\begin{aligned}
\begin{array}{l}
P, P_1, P_2 \in \schemant \\
x, y\in\schemapv,\, \vec{z} \in \schemapv^r \\
\sim \,\in \set{=, \neq}
\end{array}
\end{aligned}
\rule{0cm}{0.8cm} }
\end{align*}

An \emph{input grammar} $\gram$ is any finite subset of the schema
$\schema$, and it defines a set of programs, denoted $L(\gram)$. We
can now define the main problem addressed in this work.

\begin{definition}[Uninterpreted Program Realizability and Synthesis]
\newline
  Given an input grammar $\gram$, the realizability problem is to
  determine whether there is an uninterpreted program $p \in L(\gram)$
  such that $p$ is correct. The synthesis problem is to determine the
  above, and further, if realizable, synthesize a correct program
  $p \in L(\gram)$.
\end{definition}

\begin{example}
  Consider the program with a hole from~\exref{crypto-ex}
  (\figref{motivating-examples}, left). We can model that synthesis
  problem in our framework with the following grammar.
\begin{align*}
&\begin{array}{ll}
S   \produces P_1; P_2; P_{\hole{}}; P_3 & P_{\hole{}} \produces \stmt_{V_{\hole{}}} \\
P_1 \produces \dblqt{\cd{cipher} \passign \cd{enc(secret, key)}} &
P_3 \produces \dblqt{\passert (\cd{z} = \cd{secret})} \\
P_2 \produces \dblqt{\passume (\cd{secret} = \cd{dec(cipher, key)})} &
\end{array}
\end{align*}
Here, $V_{\hole{}} = \set{\cd{cipher}, \cd{y}_1, \ldots, \cd{y}_n}$
and the grammar $\stmt_{V_{\hole{}}}$ is that of~\secref{syntax},
restricted to loop-free programs.  Any program generated from this
grammar indeed matches the template from~\figref{motivating-examples}
(left) and any such program is correct if it satisfies the last
assertion for all models, i.e., all interpretations of the function
symbols \cd{enc} and \cd{dec} and for all initial values of the
variables in $V = V_{\hole{}} \cup \set{\cd{key}, \cd{secret}}$.
\end{example}


\section{Undecidability of Uninterpreted Program Synthesis}
\seclabel{undec}

Since verification of uninterpreted programs with loops is
undecidable~\cite{MMV19,Muller2005herbrand}, the following is
immediate.

\begin{theorem}
\thmlabel{undec}
The uninterpreted program synthesis problem is undecidable.
\end{theorem}
We next consider synthesizing loop-free uninterpreted programs (for
which verification reduces to satisfiability of quantifier-free EUF)
from grammars conforming to the following schema:

\vspace{-0.15in}
\[\schema_\text{loop-free} =
\schema \setminus
\setpred{\dblqt{P \produces \, \pwhile \, (x \sim y) \,\, P_1}}{
  P, P_1 \in \schemant, \
x, y\in \schemapv,
\sim \,\in \set{=, \neq}
}
\]
\begin{theorem}
  \thmlabel{undec-loop-free} The uninterpreted program synthesis
  problem is undecidable for the schema $\schema_\text{loop-free}$.
\end{theorem}

This is a corollary of the following stronger result: synthesis of
\emph{straight-line uninterpreted programs} (conforming to schema
$\schema_\slp$ below) is undecidable.
\begin{align*}
\schema_\slp = \schema_\text{loop-free} \setminus
\{
\dblqt{P \produces \, \pif (x\sim y) \, \pthen \, P_1 \,
  \pelse \, P_2} \,\,|\,\, &P, P_1, P_2 \in \schemant, \\
  &x, y \in \schemapv, \,\sim \,\in \set{=, \neq}
\}
\end{align*}


\begin{theorem}
  \thmlabel{undec-slp} The uninterpreted program synthesis problem is
  undecidable for the schema $\schema_\slp$.
\end{theorem}

In summary, program synthesis of even straight-line uninterpreted
programs, which have neither conditionals nor iteration, is already
undecidable.  The notion of \emph{coherence} for uninterpreted
programs was shown to yield decidable verification in~\cite{MMV19}.
As we'll see in \secref{coherent-synth}, restricting to coherent
programs yields decidable synthesis, even for programs with
conditionals \emph{and} iteration.



\section{Synthesis of Coherent Uninterpreted Programs}
\seclabel{coherent-synth}

In this section, we present the main result of the paper:
grammar-restricted program synthesis for uninterpreted \emph{coherent}
programs~\cite{MMV19} is decidable. Intuitively, coherence allows us
to maintain congruence closure in a streaming fashion when reading a
coherent execution. 
First we recall the definition of coherent executions and programs
in~\secref{coherence} and also the algorithm for verification of such
programs. Then we introduce the synthesis procedure, which works by
constructing a two-way alternating tree automaton. We briefly discuss
this class of tree automata in~\secref{tree-aut} and recall some
standard results.  In
Sections~\ref{sec:grammar-to-aut}-\ref{sec:overall-algo} we describe
the details of the synthesis procedure, argue its correctness, and
discuss its complexity.  In \subsecref{lower-bound}, we present a
tight lower bound result.

\subsection{Coherent Executions and Programs}
\seclabel{coherence}

The notion of coherence for an execution $\pi$ is defined with respect
to the \emph{terms} it computes.  Intuitively, at the beginning of an
execution, each variable $x \in V$ stores some constant term
$\init{x} \in \Cc$.  As the execution proceeds, new terms are computed
and stored in variables.  Let $\Terms_\Sigma$ be the set of all ground
terms defined using the constants and functions in $\Sigma$.
Formally, the term corresponding to a variable $x \in V$ at the end of
an execution $\pi \in \Pi_V^*$, denoted
$\comp(\pi, x) \in \Terms_\Sigma$, is inductively defined as follows.
We assume that the set of constants $\Cc$ includes a designated set of
\emph{initial} constants
$\init{V} = \setpred{\init{x}}{x \in V} \subseteq \Cc$.
\begin{align*}
\begin{array}{rclr}
\comp(\varepsilon, x) &=& \init{x} & x \in V \\
\comp(\pi{\cdot}\dblqt{x \passign y}, x) &=& \comp(\pi, y) & x, y \in V \\
\comp(\pi{\cdot}\dblqt{x \passign f(z_1, \ldots, z_r)}, x) &=&
                                                               f(\comp(\pi, z_1), \ldots, \comp(\pi, z_r)) \quad & x, z_1, \ldots, z_r \in V \\
\comp(\pi{\cdot}a, x) &=& \comp(\pi, x) & \text{otherwise} \\
\end{array}
\end{align*}
We will use $\comp(\pi)$ to denote the set
$\setpred{\comp(\pi', x)}{x \in V, \pi' \text{ is a
    prefix of } \pi}$.

A related notion is the set of \emph{term equality assumptions} that
an execution accumulates, which we formalize as
$\alpha : \pi \to \Pp(\Terms_\Sigma \times \Terms_\Sigma)$,
and define inductively as
$\alpha(\varepsilon) = \emptyset$,
$\alpha(\pi{\cdot}\dblqt{\passume(x = y)}) = \alpha(\pi) \cup \set{(\comp(\pi, x), \comp(\pi, y))}$, and
$\alpha(\pi{\cdot}a) = \alpha(\pi)$ otherwise.

For a set of term equalities
$A \subseteq \Terms_\Sigma \times \Terms_\Sigma$, and two ground terms
$t_1, t_2 \in \Terms_\Sigma$, we say $t_1$ and $t_2$ are
\emph{equivalent modulo} $A$, denoted $t_1 \cong_A t_2$, if
$A \models t_1 = t_2$.  For a set of terms
$S \subseteq \Terms_\Sigma$, and a term $t \in \Terms_\Sigma$ we write
$t \in_{A} S$ if there is a term $t' \in S$ such that
$t \cong_{A} t'$.  For terms $t, s \in \Terms_\Sigma$, we say $s$ is a
\emph{superterm modulo} $A$ of $t$, denoted $\superterm{t}{A}{s}$ if
there are terms $t', s' \in \Terms_\Sigma$ such that $t \cong_A t'$,
$s \cong_A s'$ and $s'$ is a superterm of $t'$.

With the above notation in mind, we now review the notion of
coherence.
\begin{definition}[Coherent Executions and Programs~\cite{MMV19}]
  An execution $\pi \in \Pi_V^*$ is said to be \emph{coherent} if it
  satisfies the following two conditions.
\begin{description}
	\item[\quad Memoizing.] Let $\rho = \rho' \cdot \dblqt{x \passign f(\vec{y})}$
	be a prefix of $\pi$.
	If $t_x = \comp(\rho, x) \in_{\alpha(\rho')} \comp(\rho')$,
	then there is a variable $z \in V$ such that
	$t_x \cong_{\alpha(\rho')} t_z$, where $t_z = \comp(\rho', z)$.
	\item[\quad Early Assumes.] Let $\rho = \rho' \cdot \dblqt{\passume (x = y)}$
	be a prefix of $\pi$, $t_x = \comp(\rho', x)$ and $t_y = \comp(\rho', y)$.
	If there is a term $s \in \comp(\rho')$
	such that either $\superterm{t_x}{\alpha(\rho')}{s}$
	or  $\superterm{t_y}{\alpha(\rho')}{s}$,
	then there is a variable $z \in V$
	such that $s \cong_{\alpha(\rho')} t_z$, where $t_z = \comp(\rho', z)$.
\end{description}
A program $p$ is coherent if every complete execution
$\pi \in \exec(p)$ is coherent.
\end{definition}

The following theorems due to~\cite{MMV19} establish the decidability
of verifying coherent programs and also of checking if a program is
coherent.

\begin{theorem}[\cite{MMV19}]
  \thmlabel{coherent-ver} The verification problem for coherent
  programs, i.e. checking if a given uninterpreted coherent program is
  correct, is decidable.
\end{theorem}

\begin{theorem}[\cite{MMV19}]
  \thmlabel{coherent-check} The problem of checking coherence, i.e.
  checking if a given uninterpreted program is coherent, is decidable.
\end{theorem}

The techniques used in~\cite{MMV19} are automata theoretic. They allow
us to construct an automaton $\wordaut{\Aa}_\cohcorexec$\footnote{We
  use superscripts `$\word$' and `$\tree$' for word and tree automata,
  respectively.}, of size $O(2^{\text{poly}(|V|)})$, which accepts all
coherent executions that are also correct.

To give some intuition for the notion of coherence,
we illustrate simple example programs that are not coherent. Consider program $p_\ctr$ below, which is not coherent because it fails to be memoizing.\\
\centerline{$
  p_\ctr \quad \delequal \quad \cd{x} \passign \cd{f(y)};\, \cd{x}
  \passign \cd{f(x)};\, \cd{z} \passign \cd{f(y)}
$
}
The first and third statements compute $f(\init{y})$, storing it in
variables $x$ and $z$, respectively, but the term is \emph{dropped}
after the second statement and hence is not contained in any program variable
\incctr when the third statement executes.  Next consider program
$p_\ctr$, which is not coherent because it fails to have early
assumes.  \centerline{$
  p_\ctr \quad \delequal \quad \cd{x} \passign \cd{f(w)};\, \cd{x}
  \passign \cd{f(x)}; \cd{y} \passign \cd{f(z)};\, \cd{y} \passign
  \cd{f(y)};\, \passume\cd{(w = z)}
  $ } Indeed, the assume statement is not early because superterms of
$w$ and $z$, namely $f(\init{w})$ and $f(\init{z})$, were computed and
subsequently dropped before the assume.

Intuitively, the coherence conditions are necessary to allow equality
information to be tracked with finite memory. We can make this stark
by tweaking the example for $p_1$ above as follows.
\vspace{-0.1in}
\begin{align*}
  p'_\ctr \quad \delequal \quad &\cd{x} \passign \cd{f(w)};\,\underbrace{\cd{x}\passign\cd{f(x)}\cdots\cd{x}\passign\cd{f(x)}}_{n\,\, \text{times}}; \\ &\cd{y} \passign \cd{f(z)};\,\underbrace{\cd{y}\passign\cd{f(y)}\cdots\cd{y}\passign\cd{f(y)}}_{n\,\, \text{times}};\, \passume\cd{(w = z)}
\end{align*}
Observe that, for large $n$ (e.g. $n>100$), many terms are computed
and dropped by this program, like $f^{42}(\init{x})$ and
$f^{99}(\init{y})$ for instance. The difficulty with this program,
from a verification perspective, is that the assume statement entails
equalities between many terms which have not been kept track
of. Imagine trying to verify the following program
\begin{align*}
  p_2 \quad \delequal \quad p'_1;\, \passert\cd{(x = y)}
\end{align*}
Let $\pi_{p'_1}\in \exec(p'_1)$ be the unique complete execution of
$p'_1$. If we examine the details, we see that
$t_x = \comp(\pi_{p'_1},x) = f^{101}(\init{w})$ and
$t_y = \comp(\pi_{p'_1},y) = f^{101}(\init{z})$. The assertion indeed
holds because $t_x \cong_{\{(\init{w},\init{z})\}} t_y$. However, to
keep track of this fact requires remembering an arbitrary number of
terms that grows with the size of the program. Finally, we note that
the coherence restriction is met by many single-pass algorithms,
e.g. searching and manipulation of lists and trees.


\subsection{Overview of the Synthesis Procedure}
\seclabel{tree-aut}

Our synthesis procedure uses tree automata.  We consider tree
representations of programs, or \emph{program trees}. The synthesis
problem is thus to check if there is a program tree whose
corresponding program is coherent, correct, and belongs to the input
grammar $\gram$.

The synthesis procedure works as follows.  We first construct a
top-down tree automaton $\treeaut{\Aa}_\gram$ that accepts the set of
trees corresponding to the programs generated by $\gram$.  We next
construct another tree automaton $\treeaut{\Aa}_\cohcor$, which
accepts all trees corresponding to programs that are \underline{\sf
  c}oherent and \underline{\sf c}orrect.  $\treeaut{\Aa}_\cohcor$ is a
two-way alternating tree automaton that simulates all executions of an
input program tree and checks that each is both correct and
coherent. In order to simulate longer and longer executions arising
from constructs like $\pwhile$-loops, the automaton traverses the
input tree and performs multiple passes over subtrees, visiting the
internal nodes of the tree many times. We then translate the two-way
alternating tree automaton to an equivalent (one-way) nondeterministic
top-down tree automaton by adapting results
from~\cite{vardikupferman2000,Vardi98} to our setting. Finally, we
check emptiness of the intersection between this top-down automaton
and the grammar automaton $\treeaut{\Aa}_\gram$. The definitions for
trees and the relevant automata are standard, and we refer the reader
to~\cite{tata2007} and to our technical
report~\cite{techreport}. 

\subsection{Tree Automaton for Program Trees}
\seclabel{grammar-to-aut}

Every program can be represented as a tree whose leaves are labeled with basic
statements like $\dblqt{x \passign y}$ and whose internal nodes are labeled with
constructs like $\pwhile$ and $\pseq$ (an alias for the sequencing construct `$\textbf{;}$'),
which have sub-programs as children. Essentially, we represent the
set of programs generated by an input grammar $\gram$ as a regular set
of program trees, accepted by a nondeterministic top-down tree
automaton $\treeaut{A}_\gram$.  The construction of
$\treeaut{A}_\gram$ mimics the standard construction for tree automata
that accept \emph{parse trees} of context free grammars. The
formalization of this intuition is straightforward, and we refer the
reader to~\cite{techreport} for details. 
We note the following fact regarding the construction of the acceptor
of program trees from a particular grammar $\gram$.
\begin{lemma}
  $\treeaut{A}_\gram$ has size $O(|\gram|)$ and can be constructed in
  time $O(|\gram|)$.\qed
\end{lemma}

\subsection{Tree Automaton for Simulating Executions}
\seclabel{two-way-simulation}

We now discuss the construction of the two-way alternating tree
automaton $\treeaut{\Aa}_\cohcor$ that underlies our synthesis
procedure. A two-way alternating tree automaton consists of a finite
set of states and a transition function that maps tuples $(q, m, a)$
of state, incoming direction, and node labels to positive Boolean formulas
over pairs $(q',m')$ of next state and next direction. In the case of
our binary program trees, incoming directions come from
$\set{\ddir,\uldir,\urdir}$, corresponding to coming down from a
parent, and up from left and right children. Next directions come from
$\set{\udir,\ldir,\rdir}$, corresponding to going up to a parent, and
down to left and right children.

The automaton $\treeaut{\Aa}_\cohcor$ is designed to accept the set of
all program trees that correspond to correct and coherent
programs. This is achieved by ensuring that a program tree is accepted
precisely when all executions of the program it represents are
accepted by the word automaton $\wordaut{\Aa}_\cohcorexec$
(\secref{coherence}). The basic idea behind $\treeaut{\Aa}_\cohcor$ is
as follows. Given a program tree $T$ as input, $\treeaut{\Aa}_\cohcor$
traverses $T$ and explores all the executions of the associated
program.  For each execution $\sigma$, $\treeaut{\Aa}_\cohcor$ keeps
track of the state that the word automaton $\wordaut{\Aa}_\cohcorexec$
would reach after reading $\sigma$. Intuitively, an accepting run of
$\treeaut{\Aa}_\cohcor$ is one which never visits the unique rejecting
state of $\wordaut{\Aa}_\cohcorexec$ during simulation.

We now give the formal description of
$\treeaut{\Aa}_\cohcor = (Q^\cohcor, I^\cohcor, 
\delta^\cohcor_0, \delta^\cohcor_1, \delta^\cohcor_2 )$, which works
over the alphabet $\Gamma_V$ described in~\secref{grammar-to-aut}.

\noindent \paragraph{\textbf{States}}
Both the full set of states and the initial set of states for
$\treeaut{\Aa}_\cohcor$ coincide with those of the word automaton
$\wordaut{\Aa}_\cohcorexec$.  That is, $Q^\cohcor = Q^\cohcorexec$ and
$I^\cohcor = \set{q^\cohcorexec_0}$, where $q^\cohcorexec_0$ is the
unique starting state of $\wordaut{\Aa}_\cohcorexec$.

\vspace{-0.07in}
\noindent\paragraph{\textbf{Transitions}}


For intuition, consider the case when the automaton's control is in
state $q$ reading an internal tree node $n$ with one child and which
is labeled by $a = \dblqt{\pwhile(x = y)}$. In the next step, the
automaton simultaneously performs two transitions corresponding to two
possibilities: entering the loop after assuming the guard
$\dblqt{x=y}$ to be true and exiting the loop with the guard being
false.  In the first of these simultaneous transitions, the automaton
moves to the left child $n{\cdot}\ldir$, and its state changes to
$q_1'$, where $q'_1 = \delta^\cohcorexec(q, \dblqt{\passume(x=y)})$.
In the second simultaneous transition, the automaton moves to the
parent node $n{\cdot}{\udir}$ (searching for the next statement to
execute, which follows the end of the loop) and changes its state to
$q'_2$, where
$q'_2 = \delta^\cohcorexec(q, \dblqt{\passume(x\neq y)})$.  We encode
these two possibilities as a \emph{conjunctive} transition of the
two-way alternating automaton. That is,
$\delta_1^\cohcor(q, m, a) = \big((q'_1, \ldir) \land (q'_2, \udir)
\big)$.

For every $i, m, a$, we have $\delta_i(\reject, m, a) = \pfalse$,
where $\reject$ is the unique, absorbing rejecting state of
$\wordaut{\Aa}_\cohcorexec$.  Below we describe the transitions from
all other states $q \neq \reject$.  All transitions
$\delta_i(q, m, a)$ not described below are $\pfalse$.

\paragraph{\textbf{Transitions from the root.}} At the root
node, labeled by $\dblqt{\proot}$, the automaton transitions as
follows:
\vspace{-0.15in}
  \begin{align*}
    \begin{aligned} \delta_1^\cohcor(q, m, \proot) =
      \begin{cases} (q, \ldir) \text{ if } m = \ddir \\ \texttt{true}
\,\, \text{ otherwise }
      \end{cases}
    \end{aligned}
  \end{align*}
  A two-way tree automaton starts in the configuration where $m$ is
  set to $\ddir$.  This means that in the very first step the
  automaton moves to the child node (direction $\ldir$).  If the
  automaton visits the root node in a subsequent step (marking the
  completion of an execution), then all transitions are enabled.

\paragraph{\textbf{Transitions from leaf nodes.}} For a leaf node with
label $a \in \Gamma_0$ and state $q$, the transition of the automaton
is $\delta_0^\cohcor(q, \ddir, a) = (\delta^\cohcorexec(q,a), \udir)$.
That is, when the automaton visits a leaf node from the parent, it
simulates reading $a$ in $\wordaut{\Aa}_\cohcorexec$ and moves to the
resulting state in the parent node.

\paragraph{\textbf{Transitions from ${\normalfont \dblqt{\pwhile}}$
    nodes.}}  As described earlier, when reading a node labeled by
$\dblqt{\pwhile(x\sim y)}$, where $\sim\, \in \set{=, \neq}$, the
automaton simulates both the possibility of entering the loop body as
well as the possibility of exiting the loop. This corresponds to a
conjunctive transition: \vspace{-0.05in}
\begin{align*}
  \delta_1^\cohcor(q, m, \dblqt{\pwhile(x\sim y)}) &= (q', \ldir \big) \land \big(q'', \udir)\\
  \textit{where~} q' &= \delta^\cohcorexec(q, \dblqt{\passume(x \sim y)})\\
  \textit{and~} q'' &= \delta^\cohcorexec(q, \dblqt{\passume(x \not\sim y)})
\end{align*}
Above, $\not\sim$ refers to $\dblqt{=}$ when $\sim$ is $\dblqt{\neq}$,
and vice versa. The first conjunct corresponds to the execution where
the program enters the loop body (assuming the guard is true), and
thus control moves to the left child of the current node, which
corresponds to the loop body. The second conjunct corresponds to the
execution where the loop guard is false and the automaton moves to the
parent of the current tree node.  Notice that, in both the conjuncts
above, the direction in which the tree automaton moves does not depend
on the last move $m$ of the state.  That is, no matter how the program
arrives at a $\pwhile$ statement, the automaton simulates both the
possibilities of entering or exiting the loop body.

\paragraph{\textbf{Transitions from
    ${\normalfont \dblqt{\pite}}$ nodes.}}  At a node labeled
$\dblqt{\pite(x \sim y)}$, when coming down the tree from the
parent, the automaton simulates both branches of the conditional:
\vspace{-0.1in}
\begin{align*}
  \begin{array}{rl}
    \delta_2^\cohcor(q, D, \dblqt{\pite(x \sim y)}) &=
                                                      (q', \ldir)
                                                      \land
                                                      (q'', \rdir) \\
    \textit{where~} q' &= \delta^\cohcorexec(q, \dblqt{\passume(x \sim y)}) \\
    \textit{and~} q'' &= \delta^\cohcorexec(q, \dblqt{\passume(x \not\sim y)})
  \end{array}
\end{align*}
The first conjunct in the transition corresponds to simulating the
word automaton on the condition $x \sim y$ and moving to the left
child, i.e. the body of the $\pthen$ branch.  Similarly, the second
conjunct corresponds to simulating the word automaton on the
negation of the condition and moving to the right child, i.e. the
body of the $\pelse$ branch.

\medskip Now consider the case when the automaton moves \emph{up} to
an $\pite$ node from a child node. In this case, the automaton moves
up to the parent node (having completed simulation of the $\pthen$ or
$\pelse$ branch) and the state $q$ remains unchanged: \vspace{-0.1in}
\begin{align*}
  \begin{array}{rclcc}
    \delta_2^\cohcor(q, m, \dblqt{\pite(x \sim y)}) & = &(q, \udir) & \quad& \quad m \in \set{\uldir, \urdir} \\
  \end{array}
\end{align*}

\paragraph{\textbf{Transitions from
    ${\normalfont \dblqt{\pseq}}$ nodes.}} In this case, the
automaton moves either to the left child, the right child, or to the
parent, depending on the last move.  It does not change the state
component.  Formally,
\begin{align*}
  \begin{aligned}
    \delta_2^\cohcor(q, m, \dblqt{\pseq}) =
    \begin{cases}
      (q, \ldir) \text{ if } m = \ddir \\
      (q, \rdir) \text{ if } m = \uldir \\
      (q, \udir) \text{ if } m = \urdir
    \end{cases}
  \end{aligned}
\end{align*}
The above transitions match the straightforward semantics of
sequencing two statements $s_1; s_2$.  If the automaton visits from
the parent node, it next moves to the left child to simulate $s_1$.
When it finishes simulating $s_1$, it comes up from the left child and
enters the right child to begin simulating $s_2$.  Finally, when
simulation of $s_2$ is complete, the automaton moves to the parent
node, exiting the subtree.

The following lemma asserts the correctness of the
automaton construction and states its complexity.
\begin{lemma}
  \lemlabel{tree-aut-sim-correct} $\treeaut{\Aa}_\cohcor$ accepts the
  set of all program trees corresponding to correct, coherent
  programs. 
  It has size $|\treeaut{A}_\cohcor| = O(2^{\emph{poly}(|V|)})$, and
  can be constructed in $O(2^{\emph{poly}(|V|)})$ time.\qed
\end{lemma}

\subsection{Synthesis Procedure}
\seclabel{overall-algo}

The rest of the synthesis procedure goes as follows.  We first
construct a nondeterministic \underline{\sf t}op-\underline{\sf d}own tree automaton
$\treeaut{\Aa}_{\cohcortopdown}$ such that
$L(\treeaut{\Aa}_{\cohcortopdown}) = L(\treeaut{\Aa}_\cohcor)$.  An
adaptation of results from~\cite{vardikupferman2000,Vardi98} ensures
that $\treeaut{\Aa}_{\cohcortopdown}$ has size
$|\treeaut{\Aa}_{\cohcortopdown}| = O(2^{2^{\text{poly}(|V|)}})$ and
can be constructed in time $O(2^{2^{\text{poly}(|V|)}})$.  Next we
construct a top-down nondeterministic tree automaton $\treeaut{\Aa}$
such that
$L(\treeaut{\Aa}) = L(\treeaut{\Aa}_{\cohcortopdown}) \cap
L(\treeaut{\Aa}_\gram) = L(\treeaut{\Aa}_\cohcor) \cap
L(\treeaut{\Aa}_\gram)$, with size
$|\treeaut{\Aa}| = O(2^{2^{\text{poly}(|V|)}}\cdot|\gram|)$ and in
time
$O(|\treeaut{\Aa}_{\cohcortopdown}|\cdot|\treeaut{\Aa}_\gram|) =
O(2^{2^{\text{poly}(|V|)}}\cdot|\gram|)$.  Finally, checking emptiness
of $\treeaut{\Aa}$ can be done in time $\ordernotation$. If non-empty,
a program tree can be constructed.

This gives us the central upper bound result of the paper.
\begin{theorem}
  \thmlabel{coh-prog-synth-dec} The grammar-restricted synthesis
  problem for uninterpreted coherent programs is decidable in
  $\twoexptime$, and in particular, in time doubly exponential in the
  number of variables and \complexitygrammar~in the size of the input
  grammar.  Furthermore, a tree automaton representing the set of
  \emph{all} correct coherent programs that conform to the grammar can
  be constructed in the same time. \qed
\end{theorem}



\subsection{Matching Lower Bound}
\subseclabel{lower-bound}

Our synthesis procedure is optimal. We prove a $\twoexptime$ lower
bound for the synthesis problem by reduction from the
$\twoexptime$-hard\textbf{} acceptance problem of \emph{alternating}
Turing machines (ATMs) with exponential space
bound~\cite{Chandra1981}.
Full details of the reduction can be found in~\cite{techreport}.

\begin{theorem}
  \thmlabel{lower-bound} The grammar-restricted synthesis problem for
  coherent uninterpreted programs is $\twoexptime$-hard.
\end{theorem}




\section{Further Results}
\seclabel{further-results}

In this section, we give results for variants of uninterpreted program
synthesis in terms of transition systems, Boolean programs, and
recursive programs.

\subsection{Synthesizing Transition Systems}
\subseclabel{trans-sys}

Here, rather than synthesizing programs from grammars, we consider
instead the synthesis of transition systems whose executions must
belong to a regular set. Our main result is that the synthesis problem
in this case is $\exptime$-complete, in contrast to grammar-restricted
program synthesis which is $\twoexptime$-complete.


\subsubsection{Transition System Definition and Semantics.}
Let us fix a set of program variables $V$ as before.
We consider the following finite alphabet
\begin{align*}
\Sigma_V = \setpred{
  \dblqt{x \passign y},
  \dblqt{x \passign f(\vec{z})},
  \dblqt{\passert(\tsfalse)},
  \dblqt{\pcheck(x=y)}
}{
  x, y, \in V,\,\vec{z} \in V^r
}
\end{align*}
Let us define $\Gamma_V \subseteq \Sigma_V$ to be the set of all
elements of the form $\dblqt{\pcheck (x=y)}$, where $x,y \in V$.  We
refer to the elements of $\Gamma_V$ as \emph{check} letters.

A (deterministic) transition system $TS$ over $V$ is a tuple
$(Q, q_0, H, \labts, \delta)$, where $Q$ is a finite set of states,
$q_0 \in Q$ is the initial state, $H \subseteq Q$ is the set of
halting states, $\labts: Q \rightarrow \Sigma_V$ is a labeling
function such that for any $q \in Q$, if
$\labts(q) = \dblqt{\passert(\tsfalse)}$ then $q \in H$, and
$\delta: (Q \setminus H) \rightarrow Q \cup (Q \times Q)$ is a
transition function such that for any $q \in Q \setminus H$,
$\delta(q) \in Q \times Q$ iff $\labts(q) \in \Gamma_\textit{V}$.


We define the semantics of a transition system using the set of
executions that it generates.  A \emph{(partial) execution} $\pi$ of a
transition system $TS=(Q, q_0, H, \labts, \delta)$ over variables $V$
is a finite word over the induced execution alphabet $\Pi_{V}$ (from
\secref{prelim}) with the following property. If
$\pi = a_0 a_1 \ldots a_n$ with $n \geq 0$, then there exists a
sequence of states $q_{j_0}, q_{j_1}, \ldots, q_{j_n}$ with
$q_{j_0}=q_0$ such that ($0 \leq i \leq n$):
\begin{itemize}
\item If $\labts(q_{j_i})\notin\Gamma_V$ then $a_i = \labts(q_{j_i})$,
  and if $i < n$ then $q_{j_{i+1}} = \delta(q_{j_i})$.
\item Otherwise
      $
      \begin{aligned}
\begin{cases}
        \text{either } &a_i = \dblqt{\passume(x=y)} \text{ and }
        i<n\Rightarrow q_{j_{i+1}} =
        \delta(q_{j_i})\downharpoonright_1,
        \\
        \text{or } &a_i = \dblqt{\passume(x \neq y)} \text{ and }
        i<n\Rightarrow q_{j_{i+1}} = \delta(q_{j_i})\downharpoonright_2
      \end{cases}
      \end{aligned}
$
\end{itemize}


In the above, we denote pair projection with $\downharpoonright$,
i.e., $(t_1, t_2) \downharpoonright_i = t_i$, where $i \in \{1,2\}$. A
\emph{complete execution} is an execution whose corresponding final
state ($q_n$ above) is in $H$.
For any transition system $TS$, we denote the set of its executions by
$\tsexec(TS)$ and the set of its complete executions by
$\tscexec(TS)$. The notions of \emph{correctness} and \emph{coherence}
for transition systems are identical to their counterparts for
programs. 



\vspace{-0.15in}
\subsubsection{The Transition System Synthesis Problem.} We consider
transition system specifications that place restrictions on executions
(both partial and complete) using two regular languages $S$ and $R$.
Executions must belong to the first language $S$ (which is
prefix-closed) and all complete executions must belong to the second
language $R$. A specification is given as two deterministic automata
$\wordaut{\Aa}_S$ and $\wordaut{\Aa}_R$ over executions, where
$L(\wordaut{\Aa}_S) = S$ and $L(\wordaut{\Aa}_R) = R$. For a
transition system $TS$ and specification automata $\wordaut{\Aa}_S$
and $\wordaut{\Aa}_R$, whenever
$\tsexec(TS) \subseteq L(\wordaut{\Aa}_S)$ and
$\tscexec(TS) \subseteq L(\wordaut{\Aa}_R)$ we say that $TS$ satisfies
its (syntactic) specification. Note that this need not entail
correctness of $TS$. Splitting the specification into partial
executions $S$ and complete executions $R$ allows us, among other
things, to constrain the executions of non-halting transition systems.
\medskip

\begin{definition}[Transition System Realizability and Synthesis]
  Given a finite set of program variables $V$ and deterministic
  specification automata $\wordaut{\Aa}_S$ (prefix-closed) and
  $\wordaut{\Aa}_R$ over the execution alphabet $\Pi_V$, decide if
  there is a \emph{correct}, coherent transition system $TS$ over $V$
  that satisfies the specification. Furthermore, produce one if it
  exists.
\end{definition}

Since programs are readily translated to transition systems (of
similar size), the transition system synthesis problem seems, at first
glance, to be a problem that ought to have similar
complexity. However, as we show, it is crucially different in that it
allows the synthesized transition system to have \emph{complete
  information} of past commands executed at any point. We will observe
in this section that the transition system synthesis problem is
$\exptime$-complete.

To see the difference between program and transition system synthesis,
consider program skeleton $P$ from~\exref{incomplete-info-ex}
in~\secref{examples}.
The problem is to fill the hole in $P$ with either
$\code{y} \passign \code{T}$ or $\code{y} \passign \code{F}$. Observe
that when $P$ executes, there are \emph{two} different executions that
lead to the hole. In grammar-restricted program synthesis, the hole
must be filled by a sub-program that is executed \emph{no matter how
  the hole is reached}, and hence no such program exists. However,
when we model this problem in the setting of transition systems, the
synthesizer is able to produce transitions that depend on how the hole
is reached. In other words, it does not fill the hole in $P$ with
\emph{uniform} code.
In this sense, in grammar-restricted program synthesis, programs have
\emph{incomplete information} of the past.
We crucially exploited this difference in the proof of
$\twoexptime$-hardness for grammar-restricted program synthesis (see~\cite{techreport}).
No such
incomplete information can be enforced by regular execution
specifications in transition system synthesis, and indeed the problem
turns out to be easier: transition system realizability and synthesis
are $\exptime$-complete.

\begin{theorem}
  \thmlabel{transition-system-result} Transition system realizability
  is decidable in time exponential in the number of program variables
  and polynomial in the size of the automata $\wordaut{\Aa}_S$ and
  $\wordaut{\Aa}_R$.  Furthermore, the problem is
  $\exptime$-complete. When realizable, within the same time bounds we
  can construct a correct, coherent transition system whose partial
  and complete executions are in $L(\wordaut{\Aa}_S)$ and
  $L(\wordaut{\Aa}_R)$, respectively.
\end{theorem}


\subsection{Synthesizing Boolean Programs}
\subseclabel{bool-programs} Here we observe corollaries of our results
when applied to the more restricted problem of synthesizing Boolean
programs.

In Boolean program synthesis we interpret variables in programs over
the Boolean domain $\{T, F\}$, and we disallow computations of
uninterpreted functions and the checking of uninterpreted
relations. Standard Boolean functions like $\wedge$ and $\neg$ are
instead allowed, but note that these can be modeled using conditional
statements.  We allow for \emph{nondeterminism} with a special
assignment $\dblqt{b \passign \cd{*}}$, which assigns $b$
nondeterministically to $T$ or $F$. As usual, a program is correct
when it satisfies all its assertions.

Synthesis of Boolean programs can be reduced to uninterpreted program
synthesis using two special constants $T$ and $F$. Each
nondeterministic assignment is modeled by computing a $\cd{next}$
function on successive nodes of a linked list, accessing a
nondeterministic value by computing $\cd{key}$ on the current node,
and assuming the result is either $T$ or $F$. Since uninterpreted
programs must satisfy assertions in all models, this indeed captures
nondeterministic assignment.
Further, every term ever computed in such a program is equivalent to
$T$ or $F$ (by virtue of the interleaved $\passume$ statements),
making the resulting program coherent.
The $\twoexptime$ upper bound for Boolean
program synthesis now follows from~\thmref{coh-prog-synth-dec}.
We further show that, perhaps surprisingly, the $\twoexptime$ lower bound
from~\secref{coherent-synth} can be adapted to prove
$\twoexptime$-hardness of Boolean program synthesis.


\begin{theorem}
  The grammar-restricted synthesis problem for Boolean programs is
  $\twoexptime$-complete, and can be solved in time doubly-exponential
  in the number of variables and \complexitygrammar~in the size of the
  input grammar.\qed
\end{theorem}

Thus synthesis for coherent uninterpreted programs is no more complex
than Boolean program synthesis, establishing decidability and
complexity of a problem which has found wide use in practice---for
instance, the synthesis tool {\sc Sketch} solves precisely this
problem, as it models integers using a small number of bits and allows
grammars to restrict programs with holes.


\subsection{Synthesizing Recursive Programs}
\subseclabel{recursive-pgm}

\newcommand{\trcc} {\treeaut{\Aa}_{\text{rcc}}}
\newcommand{\wrcc} {\wordaut{\Aa}_{\text{rcc}}}


We extend the positive result of~\secref{coherent-synth} to synthesize
coherent recursive programs. The setup for the problem is very
similar. Given a grammar that identifies a class of recursive
programs, the goal is to determine if there is a program in the
grammar that is coherent and correct.

The syntax of recursive programs is similar to the non-recursive case,
and we refer the reader to~\cite{techreport} for details. In essence,
programs are extended with a new function call construct. Proofs are
similar in structure to the non-recursive case, with the added
challenge of needing to account for recursive function calls and the
fact that $\wordaut{\Aa}_\cohcorexec$ becomes a (visibly) pushdown
automaton rather than a standard finite automaton. This gives a
$\twoexptime$ algorithm for synthesizing recursive programs; a
matching lower bound follows from the non-recursive case.
\seclabel{lower-bound}

\vspace{-.05in}
\begin{theorem}
  \thmlabel{rec-synthesis} The grammar-restricted synthesis problem
  for uninterpreted coherent \emph{recursive} programs is
  $\twoexptime$-complete. The algorithm is doubly exponential in the
  number of program variables and \complexitygrammar~in the size of
  the input grammar. Furthermore, a tree automaton representing the
  set of all correct, coherent recursive programs that conform to the
  grammar can be constructed in the same time.
\end{theorem}



\section{Related Work}
\seclabel{related}

The automata and game-theoretic approaches to synthesis date back to a
problem proposed by Church~\cite{church60}, after which a rich theory
emerged~\cite{BuchiLandweber69,Rabin72,automata-logics-games,kpvPneuli}.  The
problems considered in this line of work typically deal with a system
reacting to an environment interactively using a finite set of signals
over an infinite number of rounds. Tree automata over infinite
\emph{trees}, representing strategies, with various infinitary
acceptance conditions (B\"uchi, Rabin, Muller, parity) emerged as a
uniform technique to solve such synthesis problems against temporal
logic specifications with optimal complexity
bounds~\cite{PR89,KMTV00,PR90,madhudistsynth}. In this paper, we use
an alternative approach from~\cite{csl11} that works on \emph{finite}
program trees, using two-way traversals to simulate iteration. The
work in~\cite{csl11}, however, uses such representations to solve
synthesis problems for programs over a fixed finite set of Boolean
variables and against LTL specifications. In this work we use it to
synthesize coherent programs that have finitely many variables working
over infinite domains endowed with functions and relations.

While decidability results for program synthesis beyond finite data
domains are uncommon, we do know of some results of this kind. First,
there are decidability results known for synthesis of tranducers with
registers~\cite{khalimov18}. Transducers interactively read a stream
of inputs and emit a stream of outputs.  Finite-state tranducers can
be endowed with a set of registers for storing inputs and doing only
equality/disequality comparisons on future inputs. Synthesis of such
transducers for temporal logic specifications is known to be
decidable. Note that, although the data domain is infinite, there are
no functions or relations on data (other than equality), making it a
much more restricted class (and grammar-based approaches for
syntactically restricting transducers has not been studied). Indeed,
with uninterpreted functions and relations, the synthesis problem is
undecidable (\thmref{undec}), with decidability only for coherent
programs. In~\cite{caulfieldarxiv}, the authors study the problem of
synthesizing uninterpreted terms from a grammar that satisfy a
first-order specification. They give various decidability and
undecidability results. In contrast, our results are for programs with
conditionals and iteration (but restricted to coherent programs) and
for specifications using assertions in code.

Another setting with a decidable synthesis result over unbounded
domains is work on strategy synthesis for linear arithmetic
\emph{satisfiability} games~\cite{kincaid18}. There it is shown that
for a satisfiability game, in which two players (SAT and UNSAT) play
to prove a formula is satisfiable (where the formula is interpreted
over the theory of linear rational arithmetic), if the SAT player has
a winning strategy then a strategy can be synthesized. Though the data
domain (rationals) is infinite, the game consists of a finite set of
interactions and hence has no need for recursion. The authors also
consider reachability games where the number of rounds can be
unbounded, but present only sound and incomplete results, as checking
who wins in such reachability games is undecidable.

Tree automata techniques for accepting finite parse trees of programs
was explored in~\cite{madhu2011} for synthesizing reactive programs
with variables over finite domains.  In more recent work, automata on
finite trees have been explored for synthesizing data completion
scripts from input-output examples~\cite{WangGulwaniSinghOOPSLA2016},
for accepting programs that are verifiable using abstract
interpretations~\cite{Wang2017}, and for relational program
synthesis~\cite{WangWangDilligOOPSLA18}.

The work in~\cite{MMSV2018} explores a decidable logic with
$\exists^* \forall^*$ prefixes that can be used to encode synthesis
problems with background theories like arithmetic.  However, encoding
program synthesis in this logic only expresses programs of finite
size.  Another recent paper~\cite{HuUnrealizability2019} explores
sound (but incomplete) techniques for showing unrealizability of
syntax-guided synthesis problems.



\section{Conclusions}
\seclabel{conclusions}


We presented foundational results on synthesizing coherent programs
with uninterpreted functions and relations.  To the best of our
knowledge, this is the first natural decidable program synthesis
problem for programs of arbitrary size which have iteration/recursion,
and which work over infinite domains.

The field of program synthesis lacks theoretical results, and
especially decidability results. We believe our results to be the
first of their kind to fill this lacuna, and we find this paper
exciting because it bridges the worlds of program synthesis and the
rich classical synthesis frameworks of systems over finite domains
using tree
automata~\cite{BuchiLandweber69,Rabin72,automata-logics-games,kpvPneuli}.  We
believe this link could revitalize both domains with new techniques
and applications.

Turning to practical applications of our results, several questions
require exploration in future work. First, one might question the
utility of programs that verify only with respect to uninterpreted
data domains. Recent work~\cite{euforia2019} has shown that verifying
programs using uninterpreted abstractions can be extremely effective
in practice for proving programs correct. Also, recent work by Mathur
et al.~\cite{MMV20Axioms} explores ways to add \emph{axioms} (such as
commutativity of functions, axioms regarding partial orders, etc.) and
yet preserve decidability of verification. The methods used therein
are compatible with our technique, and we believe our results can be
extended smoothly to their decidable settings. A more elaborate way to
bring in complex theories (like arithmetic) would be to marry our
technique with the \emph{iterative} automata-based software
verification technique pioneered by work behind the {\sc Ultimate}
tool~\cite{Matthias2013,Heizmann2010,Heizmann2009SAS,UltimateAutomizerTacas};
this won't yield decidable synthesis, but still could result in
\emph{complete} synthesis procedures.

The second concern for practicality is the coherence
restriction. There is recent work by Mathur et
al.~\cite{mathur2019forest} that shows single-pass heap-manipulating
programs respect a (suitably adapted) notion of coherence. Adapting
our technique to this setting seems feasible, and this would give an
interesting application of our work. Finally, it is important to build
an implementation of our procedure in a tool that exploits pragmatic
techniques for constructing tree automata, and the techniques pursued
in~\cite{WangGulwaniSinghOOPSLA2016,Wang2017,WangWangDilligOOPSLA18}
hold promise.

\newpage

\bibliographystyle{splncs04}
\bibliography{pgmsynth}

\begin{thebibliography}{10}
\providecommand{\url}[1]{\texttt{#1}}
\providecommand{\urlprefix}{URL }
\providecommand{\doi}[1]{https://doi.org/#1}

\bibitem{sygus}
Alur, R., Bod{\'{\i}}k, R., Dallal, E., Fisman, D., Garg, P., Juniwal, G.,
  Kress{-}Gazit, H., Madhusudan, P., Martin, M.M.K., Raghothaman, M., Saha, S.,
  Seshia, S.A., Singh, R., Solar{-}Lezama, A., Torlak, E., Udupa, A.:
  Syntax-guided synthesis. In: Dependable Software Systems Engineering, {NATO}
  Science for Peace and Security Series, {D:} Information and Communication
  Security, vol.~40, pp. 1--25. {IOS} Press (2015)

\bibitem{Alur2004vpa}
Alur, R., Madhusudan, P.: Visibly pushdown languages. In: Proceedings of the
  Thirty-sixth Annual ACM Symposium on Theory of Computing. pp. 202--211. STOC
  '04, ACM, New York, NY, USA (2004). \doi{10.1145/1007352.1007390},
  \url{http://doi.acm.org/10.1145/1007352.1007390}

\bibitem{Alur2009vpa}
Alur, R., Madhusudan, P.: Adding nesting structure to words. J. ACM
  \textbf{56}(3),  16:1--16:43 (May 2009). \doi{10.1145/1516512.1516518},
  \url{http://doi.acm.org/10.1145/1516512.1516518}

\bibitem{SearchBasedProgramSynthesisCACM}
Alur, R., Singh, R., Fisman, D., Solar-Lezama, A.: Search-based program
  synthesis. Commun. ACM  \textbf{61}(12),  84--93 (Nov 2018).
  \doi{10.1145/3208071}, \url{http://doi.acm.org/10.1145/3208071}

\bibitem{godelcompleteness}
Bauer{-}Mengelberg, S.: \"uber die vollst\"andigkeit des logikkalk\"uls.
  Journal of Symbolic Logic  \textbf{55}(1),  341--342 (1990).
  \doi{10.2307/2274974}

\bibitem{BGJPPW07}
Bloem, R., Galler, S.J., Jobstmann, B., Piterman, N., Pnueli, A., Weiglhofer,
  M.: Specify, compile, run: Hardware from {PSL}. Electr. Notes Theor. Comput.
  Sci.  \textbf{190}(4),  3--16 (2007)

\bibitem{BJPPS12}
Bloem, R., Jobstmann, B., Piterman, N., Pnueli, A., Sa'ar, Y.: Synthesis of
  reactive(1) designs. J. Comput. Syst. Sci.  \textbf{78}(3),  911--938 (2012)

\bibitem{calcofcomputation}
Bradley, A.R., Manna, Z.: The Calculus of Computation: Decision Procedures with
  Applications to Verification. Springer-Verlag, Berlin, Heidelberg (2007)

\bibitem{BuchiLandweber69}
Buchi, J.R., Landweber, L.H.: Solving sequential conditions by finite-state
  strategies. Transactions of the American Mathematical Society  \textbf{138},
  295--311 (1969), \url{http://www.jstor.org/stable/1994916}

\bibitem{euforia2019}
Bueno, D., Sakallah, K.A.: euforia: Complete software model checking with
  uninterpreted functions. In: Enea, C., Piskac, R. (eds.) Verification, Model
  Checking, and Abstract Interpretation. pp. 363--385. Springer International
  Publishing, Cham (2019)

\bibitem{caulfieldarxiv}
Caulfield, B., Rabe, M.N., Seshia, S.A., Tripakis, S.: What's decidable about
  syntax-guided synthesis? CoRR  \textbf{abs/1510.08393} (2015)

\bibitem{Chandra1981}
Chandra, A.K., Kozen, D.C., Stockmeyer, L.J.: Alternation. J. ACM
  \textbf{28}(1),  114--133 (Jan 1981). \doi{10.1145/322234.322243},
  \url{http://doi.acm.org/10.1145/322234.322243}

\bibitem{church60}
Church, A.: Application of recursive arithmetic to the problem of circuit
  synthesis. Summaries of talks presented at the Summer Institute for Symbolic
  Logic Cornell University, 1957, 2nd edn., Journal of Symbolic Logic
  \textbf{28}(4),  30--50. 3a--45a. (1960)

\bibitem{tata2007}
Comon, H., Dauchet, M., Gilleron, R., L\"oding, C., Jacquemard, F., Lugiez, D.,
  Tison, S., Tommasi, M.: Tree automata techniques and applications. Available
  on: \url{http://www.grappa.univ-lille3.fr/tata} (2007), release October, 12th
  2007

\bibitem{dolev1983security}
Dolev, D., Yao, A.: On the security of public key protocols. IEEE Transactions
  on information theory  \textbf{29}(2),  198--208 (1983)

\bibitem{Durgin2004}
Durgin, N., Lincoln, P., Mitchell, J., Scedrov, A.: Multiset rewriting and the
  complexity of bounded security protocols. J. Comput. Secur.  \textbf{12}(2),
  247--311 (Apr 2004), \url{http://dl.acm.org/citation.cfm?id=1017273.1017276}

\bibitem{emerson-jutla}
{Emerson}, E.A., {Jutla}, C.S.: Tree automata, mu-calculus and determinacy. In:
  [1991] Proceedings 32nd Annual Symposium of Foundations of Computer Science.
  pp. 368--377 (1991)

\bibitem{kincaid18}
Farzan, A., Kincaid, Z.: Strategy synthesis for linear arithmetic games.
  {PACMPL}  \textbf{2}({POPL}),  61:1--61:30 (2018)

\bibitem{automata-logics-games}
Gr\"{a}del, E., Thomas, W., Wilke, T. (eds.): Automata Logics, and Infinite
  Games: A Guide to Current Research. Springer-Verlag, Berlin, Heidelberg
  (2002)

\bibitem{flashfill11}
Gulwani, S.: Automating string processing in spreadsheets using input-output
  examples. In: {POPL}. pp. 317--330. {ACM} (2011)

\bibitem{flashfill12}
Gulwani, S., Harris, W.R., Singh, R.: Spreadsheet data manipulation using
  examples. Commun. {ACM}  \textbf{55}(8),  97--105 (2012)

\bibitem{ProgramSynthInductiveProg}
Gulwani, S., Hern{\'{a}}ndez{-}Orallo, J., Kitzelmann, E., Muggleton, S.H.,
  Schmid, U., Zorn, B.G.: Inductive programming meets the real world. Commun.
  {ACM}  \textbf{58}(11),  90--99 (2015)

\bibitem{ProgramSynthGPS}
Gulwani, S., Polozov, O., Singh, R.: Program synthesis. Foundations and Trends
  in Programming Languages  \textbf{4}(1-2),  1--119 (2017)

\bibitem{UltimateAutomizerTacas}
Heizmann, M., Christ, J., Dietsch, D., Ermis, E., Hoenicke, J., Lindenmann, M.,
  Nutz, A., Schilling, C., Podelski, A.: Ultimate automizer with smtinterpol.
  In: Piterman, N., Smolka, S.A. (eds.) Tools and Algorithms for the
  Construction and Analysis of Systems. pp. 641--643. Springer Berlin
  Heidelberg, Berlin, Heidelberg (2013)

\bibitem{Heizmann2009SAS}
Heizmann, M., Hoenicke, J., Podelski, A.: Refinement of trace abstraction. In:
  Proceedings of the 16th International Symposium on Static Analysis. pp.
  69--85. SAS '09, Springer-Verlag, Berlin, Heidelberg (2009)

\bibitem{Heizmann2010}
Heizmann, M., Hoenicke, J., Podelski, A.: Nested interpolants. In: Proceedings
  of the 37th Annual ACM SIGPLAN-SIGACT Symposium on Principles of Programming
  Languages. pp. 471--482. POPL '10, ACM, New York, NY, USA (2010).
  \doi{10.1145/1706299.1706353},
  \url{http://doi.acm.org/10.1145/1706299.1706353}

\bibitem{Matthias2013}
Heizmann, M., Hoenicke, J., Podelski, A.: Software model checking for people
  who love automata. In: Sharygina, N., Veith, H. (eds.) Computer Aided
  Verification. pp. 36--52. Springer Berlin Heidelberg, Berlin, Heidelberg
  (2013)

\bibitem{HuUnrealizability2019}
Hu, Q., Breck, J., Cyphert, J., D'Antoni, L., Reps, T.: Proving unrealizability
  for syntax-guided synthesis. In: Dillig, I., Tasiran, S. (eds.) Computer
  Aided Verification. pp. 335--352. Springer International Publishing, Cham
  (2019)

\bibitem{JhaSeshia17}
Jha, S., Seshia, S.A.: A theory of formal synthesis via inductive learning.
  Acta Inf.  \textbf{54}(7),  693--726 (Nov 2017).
  \doi{10.1007/s00236-017-0294-5},
  \url{https://doi.org/10.1007/s00236-017-0294-5}

\bibitem{khalimov18}
Khalimov, A., Maderbacher, B., Bloem, R.: Bounded synthesis of register
  transducers. In: Lahiri, S.K., Wang, C. (eds.) Automated Technology for
  Verification and Analysis. pp. 494--510. Springer International Publishing,
  Cham (2018)

\bibitem{techreport}
Krogmeier, P., Mathur, U., Murali, A., Madhusudan, P., Viswanathan, M.:
  Decidable synthesis of programs with uninterpreted functions. CoRR
  \textbf{abs/1910.09744} (2019), \url{http://arxiv.org/abs/1910.09744}

\bibitem{KMTV00}
Kupferman, O., Madhusudan, P., Thiagarajan, P.S., Vardi, M.Y.: Open systems in
  reactive environments: Control and synthesis. In: {CONCUR}. Lecture Notes in
  Computer Science, vol.~1877, pp. 92--107. Springer (2000)

\bibitem{kpvPneuli}
Kupferman, O., Piterman, N., Vardi, M.Y.: An Automata-Theoretic Approach to
  Infinite-State Systems, pp. 202--259. Springer Berlin Heidelberg, Berlin,
  Heidelberg (2010). \doi{10.1007/978-3-642-13754-9\_11}

\bibitem{vardikupferman2000}
Kupferman, O., Vardi, M.Y.: An automata-theoretic approach to reasoning about
  infinite-state systems. In: Emerson, E.A., Sistla, A.P. (eds.) Computer Aided
  Verification. pp. 36--52. Springer Berlin Heidelberg, Berlin, Heidelberg
  (2000)

\bibitem{LMN16}
L{\"{o}}ding, C., Madhusudan, P., Neider, D.: Abstract learning frameworks for
  synthesis. In: {TACAS}. Lecture Notes in Computer Science, vol.~9636, pp.
  167--185. Springer (2016)

\bibitem{csl11}
Madhusudan, P.: Synthesizing reactive programs. In: {CSL}. LIPIcs, vol.~12, pp.
  428--442. Schloss Dagstuhl - Leibniz-Zentrum fuer Informatik (2011)

\bibitem{MMSV2018}
Madhusudan, P., Mathur, U., Saha, S., Viswanathan, M.: {A Decidable Fragment of
  Second Order Logic With Applications to Synthesis}. In: Ghica, D., Jung, A.
  (eds.) 27th EACSL Annual Conference on Computer Science Logic (CSL 2018).
  Leibniz International Proceedings in Informatics (LIPIcs), vol.~119, pp.
  31:1--31:19. Schloss Dagstuhl--Leibniz-Zentrum fuer Informatik, Dagstuhl,
  Germany (2018). \doi{10.4230/LIPIcs.CSL.2018.31},
  \url{http://drops.dagstuhl.de/opus/volltexte/2018/9698}

\bibitem{madhu2011}
Madhusudan, P., Parlato, G.: The tree width of auxiliary storage. In:
  Proceedings of the 38th Annual ACM SIGPLAN-SIGACT Symposium on Principles of
  Programming Languages. pp. 283--294. POPL '11, ACM, New York, NY, USA (2011).
  \doi{10.1145/1926385.1926419},
  \url{http://doi.acm.org/10.1145/1926385.1926419}

\bibitem{madhudistsynth}
Madhusudan, P., Thiagarajan, P.S.: Distributed controller synthesis for local
  specifications. In: {ICALP}. Lecture Notes in Computer Science, vol.~2076,
  pp. 396--407. Springer (2001)

\bibitem{MMV19}
Mathur, U., Madhusudan, P., Viswanathan, M.: Decidable verification of
  uninterpreted programs. Proc. ACM Program. Lang.  \textbf{3}(POPL),
  46:1--46:29 (Jan 2019). \doi{10.1145/3290359},
  \url{http://doi.acm.org/10.1145/3290359}

\bibitem{MMV20Axioms}
Mathur, U., Madhusudan, P., Viswanathan, M.: What's decidable about program
  verification modulo axioms? In: Biere, A., Parker, D. (eds.) Tools and
  Algorithms for the Construction and Analysis of Systems. pp. 158--177.
  Springer International Publishing, Cham (2020)

\bibitem{mathur2019forest}
Mathur, U., Murali, A., Krogmeier, P., Madhusudan, P., Viswanathan, M.:
  Deciding memory safety for single-pass heap-manipulating programs. Proc. ACM
  Program. Lang.  \textbf{4}(POPL) (Dec 2019). \doi{10.1145/3371103},
  \url{https://doi.org/10.1145/3371103}

\bibitem{Mostowski1991GamesWF}
Mostowski, A.W.: Games with forbidden positions (1991)

\bibitem{Muller2005herbrand}
M\"{u}ller-Olm, M., R\"{u}thing, O., Seidl, H.: Checking herbrand equalities
  and beyond. In: Proceedings of the 6th International Conference on
  Verification, Model Checking, and Abstract Interpretation. pp. 79--96.
  VMCAI'05, Springer-Verlag, Berlin, Heidelberg (2005)

\bibitem{igoranca14}
Muscholl, A., Walukiewicz, I.: Distributed synthesis for acyclic architectures.
  In: {FSTTCS}. LIPIcs, vol.~29, pp. 639--651. Schloss Dagstuhl -
  Leibniz-Zentrum fuer Informatik (2014)

\bibitem{PR89}
Pnueli, A., Rosner, R.: On the synthesis of a reactive module. In: {POPL}. pp.
  179--190. {ACM} Press (1989)

\bibitem{PR90}
Pnueli, A., Rosner, R.: Distributed reactive systems are hard to synthesize.
  In: {FOCS}. pp. 746--757. {IEEE} Computer Society (1990)

\bibitem{post1946}
Post, E.L.: A variant of a recursively unsolvable problem. Bull. Amer. Math.
  Soc.  \textbf{52}(4),  264--268 (04 1946),
  \url{https://projecteuclid.org:443/euclid.bams/1183507843}

\bibitem{QS17}
Qiu, X., Solar{-}Lezama, A.: Natural synthesis of provably-correct
  data-structure manipulations. {PACMPL}  \textbf{1}({OOPSLA}),  65:1--65:28
  (2017). \doi{10.1145/3133889}

\bibitem{Rabin72}
Rabin, M.O.: Automata on Infinite Objects and Church's Problem. American
  Mathematical Society, Boston, MA, USA (1972)

\bibitem{singhrepair}
Singh, R., Gulwani, S., Solar-Lezama, A.: Automated feedback generation for
  introductory programming assignments. SIGPLAN Not.  \textbf{48}(6),  15--26
  (Jun 2013). \doi{10.1145/2499370.2462195},
  \url{http://doi.acm.org/10.1145/2499370.2462195}

\bibitem{sketch}
Solar-Lezama, A.: Program sketching. International Journal on Software Tools
  for Technology Transfer  \textbf{15}(5),  475--495 (Oct 2013).
  \doi{10.1007/s10009-012-0249-7},
  \url{https://doi.org/10.1007/s10009-012-0249-7}

\bibitem{sketching}
Solar{-}Lezama, A., Tancau, L., Bod{\'{\i}}k, R., Seshia, S.A., Saraswat, V.A.:
  Combinatorial sketching for finite programs. In: {ASPLOS}. pp. 404--415.
  {ACM} (2006)

\bibitem{syguswebsite}
SyGuS: Syntax guided synthesis (~), \url{https://sygus.org/}

\bibitem{Vardi98}
Vardi, M.Y.: Reasoning about the past with two-way automata. In: Larsen, K.G.,
  Skyum, S., Winskel, G. (eds.) Automata, Languages and Programming. pp.
  628--641. Springer Berlin Heidelberg, Berlin, Heidelberg (1998)

\bibitem{Wang2017}
Wang, X., Dillig, I., Singh, R.: Program synthesis using abstraction
  refinement. Proc. ACM Program. Lang.  \textbf{2}(POPL),  63:1--63:30 (Dec
  2017). \doi{10.1145/3158151}, \url{http://doi.acm.org/10.1145/3158151}

\bibitem{WangGulwaniSinghOOPSLA2016}
Wang, X., Gulwani, S., Singh, R.: Fidex: Filtering spreadsheet data using
  examples. In: Proceedings of the 2016 ACM SIGPLAN International Conference on
  Object-Oriented Programming, Systems, Languages, and Applications. pp.
  195--213. OOPSLA 2016, ACM, New York, NY, USA (2016).
  \doi{10.1145/2983990.2984030},
  \url{http://doi.acm.org/10.1145/2983990.2984030}

\bibitem{WangWangDilligOOPSLA18}
Wang, Y., Wang, X., Dillig, I.: Relational program synthesis. Proc. ACM
  Program. Lang.  \textbf{2}(OOPSLA),  155:1--155:27 (Oct 2018).
  \doi{10.1145/3276525}, \url{http://doi.acm.org/10.1145/3276525}

\end{thebibliography}

\newpage
\appendix

\section{Continued from \secref{examples}}
\applabel{examples-appendix}

\subsection{Encoding Input/Output Examples}
\applabel{inout-ex-appendix}

Here we give an encoding of three input/output example models for the
problem of finding if a linked list has a node with key \cd{k}.
\begin{align*}
\begin{array}{l}
\textsf{// \emph{Establishing constants}} \\
\passume(\cd{T} \neq \cd{F}); \\
\passume(\cd{k'} \neq \cd{k});\\\\ 
  \textsf{// \emph{Positive example}} \\
\cdm{x_2} \passign \cdm{next(x_1)};~ \cdm{x_3} \passign \cdm{next(x_2)};\\
 \passume(\cdm{x_1} \neq \cd{NIL});~\passume(\cdm{x_2} \neq \cd{NIL});~\passume(\cdm{x_3} \neq \cd{NIL});\\
 \passume(\cdm{next(x_3)} = \cd{NIL});\\
\passume(\cdm{key(x_1)} = \cd{k'});~ \passume(\cdm{key(x_2)} = \cd{k});~ \passume(\cdm{key(x_3)} = \cd{k'});\\
\passume(\cdm{x_{ans}} = \cd{T});\\\\ 
  \textsf{// \emph{Negative example}} \\
  \cdm{y_2} \passign \cdm{next(y_1)};~ \cdm{y_3} \passign \cdm{next(y_2)};\\
 \passume(\cdm{y_1} \neq \cd{NIL});~\passume(\cdm{y_2} \neq \cd{NIL});~\passume(\cdm{y_3} \neq \cd{NIL});\\
 \passume(\cdm{next(y_3)} = \cd{NIL});\\
\passume(\cdm{key(y_1)} = \cd{k'});~ \passume(\cdm{key(y_2)} = \cd{k'});~ \passume(\cdm{key(y_3)} = \cd{k'});\\
\passume(\cdm{y_{ans}} = \cd{F});\\\\ 
  \textsf{// \emph{Positive example}} \\
\cdm{z_2} \passign \cdm{next(z_1)};\\
 \passume(\cdm{z_1} \neq \cd{NIL});~\passume(\cdm{z_2} \neq \cd{NIL});~\passume(\cdm{next(z_2)} = \cd{NIL});
\\ \passume(\cdm{key(z_1)} = \cd{k});~ \passume(\cdm{key(z_2)} = \cd{k});\\
\passume(\cdm{z_{ans}} = \cd{T});\\
\end{array}
\end{align*}

The above program block defines three lists starting at $\cdm{x_1}$,
$\cdm{y_1}$ and $\cdm{z_1}$ respectively, having first defined
distinct Boolean constants like $\cd{T}$ and $\cd{F}$. These constants
are used to define expected answers $\cdm{x_{ans}}$, $\cdm{y_{ans}}$
and $\cdm{z_{ans}}$, for each example. 
Next, we can choose one of the above examples nondeterministically by
using a variable $\cd{ch}$, denoting nondeterministic \underline{\sf
  ch}oice. Since a data model gives an initial value to every
variable, each of the three examples above is chosen in some data
model.
\begin{align*}
\begin{array}{l}
  \textsf{// \emph{Nondeterministically choose an example}} \\
\passume(\cdm{ch_x} \neq \cdm{ch_y});~\passume(\cdm{ch_x} \neq \cdm{ch_z});~\passume(\cdm{ch_y} \neq \cdm{ch_z});\hspace{8em} \\
\passume(\cd{ch} = \cdm{ch_z} \lor \cd{ch} = \cdm{ch_y} \lor \cd{ch} = \cdm{ch_z}); \\
\pif(\cd{ch} = \cdm{ch_x})~ \pthen~ \cd{head} \passign \cdm{x_1}; ~ \cd{ans} \passign \cdm{x_{ans}};\\
\pif(\cd{ch} = \cdm{ch_y})~ \pthen~ \cd{head} \passign \cdm{y_1}; ~ \cd{ans} \passign \cdm{y_{ans}};\\
\pif(\cd{ch} = \cdm{ch_z})~ \pthen~ \cd{head} \passign \cdm{z_1}; ~ \cd{ans} \passign \cdm{z_{ans}};\\
\end{array}
\end{align*}

Lastly, we give the following template with a hole:
\begin{align*}
\begin{array}{l}
  \textsf{// \emph{Template with a hole}} \\
  \pwhile(\cd{head} \neq \cd{NIL})\hspace{25em}\\
  \;\;\;\;\holepred{}{can use \cd{key}, compare with \cd{k}, and assign to $\cdm{computed_{ans}}$ }; \\
\;\;\;\;\cd{head} \passign \cdm{next(head)};\\
\passert(\cd{ans} = \cdm{computed_{ans}})\\
\end{array}
\end{align*}

The full specification consists of a grammar that generates the
program blocks for the three examples, the nondeterministic choice,
and the template with the hole. It is easy to see that any correct
solution to the hole must be correct for all examples.


\newpage
\section{Continued from \secref{prelim}}
\applabel{executions}

The set of complete executions for a program $p$, denoted $\exec(p)$,
is defined inductively. In what follows, $c$ is of the form
$\dblqt{x = y}$ or $\dblqt{x \neq y}$, and we identify $\neg (x=y)$
and $\neg(x \not = y)$ with $x \not= y$ and $x=y$, respectively.
  \begin{align*}
    &\exec(\pskip) = \,\, \epsilon \\
    &\exec(x := y) = \,\, ``x := y" \\
    &\exec(x := f(\vec{z})) = \,\,  ``x := f(\vec{z})" \\
    &\exec(\passume(c)) = \,\, ``\passume(c)" \\
    &\exec(\passert(c)) = \,\, \dblqt{\passume(\neg c)} \cdot
                      ``\passert(\pfalse)" + \exec(\pskip) \\
    &\exec(\pif \,c\, \pthen \,s_1\, \pelse \,s_2\,) = \,\, \dblqt{\passume(c)} \cdot
                                         \exec(s_1) + \dblqt{\passume(\neg
                                         c)} \cdot \exec(s_2) \\
    &\exec(\pwhile \,c\, \{\, p \,\}) = \,\, \big(\dblqt{\passume(c)} \cdot \exec(\, p \,)\big)^*
                                  \cdot \dblqt{\passume(\neg c)} \\
    &\exec(p_1 ; p_2) = \,\, \exec(p_1) \cdot \exec(p_2)
  \end{align*}

\newpage

\section{Continued from \secref{undec}}
\applabel{slp-undec}

\subsection{Undecidability of Synthesising Straight-Line Programs}
\applabel{slp-undecidability-proof}

Here we present the proof of~\thmref{undec-slp}, which is a reduction
from Post's Correspondence Problem (PCP), defined below.

\begin{definition}[Post's Correspondence Problem]
  Let $\Gamma$ be an alphabet with at least two symbols. A problem
  instance consists of two lists of strings
  $\alpha = (\alpha_1, \ldots \alpha_n)$ and
  $\beta = (\beta_1, \ldots, \beta_n)$, with $n>0$. The instance is in
  the language iff there is a finite non-empty sequence of indices
  $i_1, i_2, \ldots i_N$ ($1\leq i_j \leq n$ for every
  $1 \leq j \leq N$) such that
\[
\alpha_{i_1} \cdot \alpha_{i_2} \cdot \ldots \cdot \alpha_{i_N}
=
\beta_{i_1} \cdot \beta_{i_2} \cdot \ldots \cdot \beta_{i_N}
\]
\end{definition}
It is a well-known result that PCP is undecidable~\cite{post1946}. The
reduction from PCP to synthesis of straight-line programs is as
follows. Given an instance of PCP $P = (\Gamma,\alpha,\beta)$ over
alphabet $\Gamma$ with lists of strings $\alpha$ and $\beta$, consider
the first order signature
$\Sigma_P = (\emptyset,\set{f_\sigma}_{\sigma \in \Gamma},\emptyset)$
and the grammar $\gram_P = (\Delta_P, St_P, NT_P, R_P)$ such that:
\begin{itemize}
    \item $\Delta_P = \set{\dblqt{x_1 \passign x_2},\,\dblqt{x_1 \passign x_3},\,\dblqt{;}} \cup \set{t_{1, \sigma}}_{\sigma \in \Gamma} \cup \set{t_{2, \sigma}}_{\sigma \in \Gamma} \cup \set{\dblqt{\passume(x_1 \neq x_2)},\dblqt{\passert(\pfalse)}}$
    where
	\begin{align*}
	    t_{1, \sigma} = \dblqt{x_1 \passign f_\sigma(x_1)} \\
		t_{2, \sigma} = \dblqt{x_2 \passign f_\sigma(x_2)}
	\end{align*}


	\item $NT_P = \set{S_P,Q_P,F} \cup \set{A_i}_{1\leq i\leq n}\cup \set{B_i}_{1\leq i\leq n}\cup \set{C_i}_{1\leq i\leq n}$ where $n$ is the length of the lists $\alpha$ and $\beta$ as given by $P$.

	\item $R_P$ is the following collection of rules:
	\begin{equation*}
	\begin{array}{rcl}
	S_P &\to& x_1 \passign x_3 \,;\, x_2 \passign x_3 \,;\, Q_P \\
	Q_P &\to& Q_P \,;\, Q_P \\
	Q_P &\to& C_1 \\
	Q_P &\to& C_2 \\
	& \vdots & \\
	Q_P &\to& C_n \\
	C_1 &\to& A_1 \,;\, B_1 \\
	& \vdots & \\
	C_n &\to& A_n \,;\, B_n \\
	F &\to& \passume(x_1 \neq x_2) \,;\, \passert(\pfalse) \\
	\end{array}
	\end{equation*}

	Let $\alpha_i = \sigma_{j_1}\sigma_{j_2}\cdots \sigma_{j_{|\alpha_i|}}$, where $\sigma_{j_k} \in \Gamma$ for every $1\leq k\leq {|\alpha_i|}$.
	Then, the production rule for $A_i$ is given by
	\begin{equation*}
	\begin{array}{rcl}
	A_i &\to& t_{1, \sigma_{j_1}}\,;\,t_{1, \sigma_{j_2}}\,;\,\ldots\,;\,t_{1, \sigma_{j_{|\alpha_i|}}}
	\end{array}
	\end{equation*}

	Similarly, let $\beta_i = \sigma_{l_1}\sigma_{l_2}\cdots \sigma_{l_{|\beta_i|}}$, where $\sigma_{l_k} \in \Gamma$ for every $1\leq l\leq {|\beta_i|}$.
	Then, the production rule for $B_i$ is given by
	\begin{equation*}
	\begin{array}{rcl}
	B_i &\to& t_{2, \sigma_{l_1}}\,;\,t_{2, \sigma_{l_2}}\,;\,\ldots\,;\,t_{2, \sigma_{l_{|\beta_i|}}}
	\end{array}
	\end{equation*}
\end{itemize}

For an intuitive understanding of the production rules for $A_i$ and
$B_i$, recall that each function in the signature is indexed by a
letter from $\Gamma$. If by abuse of notation we associate the
function $f_{\gamma_2}\circ f_{\gamma_1}$ (for
$\gamma_1,\gamma_2 \in \Gamma$) with the symbol
$f_{\gamma_1\cdot\gamma_2}$ (and similarly for longer compositions),
then the production rule for $A_i$ produces a program block that
updates the variable $x_1$ to $f_{\alpha_i}(x_1)$. Similarly $B_i$
produces a program block that updates $x_2$ to
$f_{\beta_i}(x_2)$. 

Although the grammar as presented does not quite conform to the
grammar schema $\schema_\slp$ of straight-line programs over the given
signature, it can be rewritten into an equivalent one that does
conform by introducing some extra nonterminals and folding the
productions $C_1 \ldots C_n$ into productions for $Q_P$. We claim that
the uninterpreted synthesis problem over this grammar is equivalent to
the given PCP instance.

To prove this claim, let's observe that every program generated by
this grammar is of the form
\begin{align*}
	x_1 \passign x_3 \,;\, x_2 \passign x_3 \,;\, C_{i_1} \,;\, C_{i_2}
  \,;\, C_{i_N} \,;\, \passume(x_1 \neq x_2) \,;\, \passert(\pfalse)
\end{align*}
for some $N$ and some $i_j$, with $1 \leq i_j \leq n$ for every
$1 \leq j \leq N$. Let $\pi_2$ be the prefix that excludes the last
two statements, and let $\pi_1$ be the prefix that excludes the last
statement.

Consider the correctness of this program. Observe that, using our
shorthand notation, the value of the variable $x_1$ at the end of the
program block $\pi_2$ is $f_{w_\alpha}(x_3)$ where
$w_\alpha = \alpha_{i_1} \cdot \alpha_{i_2} \cdot \ldots \cdot
\alpha_{i_N}$. More precisely, in any first order model $M$ over our
signature, the value of $x_1$ at the end of $\pi_2$ is the value
(given by $M$) corresponding to the term $f_{w_\alpha}(\init{x_3})$
where by $\init{x_3}$ we mean the initial value of the variable $x_3$
(which can be modelled with an extra immutable variable). Similarly,
at the end of $\pi_2$ the value of variable $x_2$ is
$f_{w_\beta}(\init{x_3})$ where
$w_\beta = \beta_{i_1} \cdot \beta_{i_2} \cdot \ldots \cdot
\beta_{i_N}$.

For the program to be correct, the prefix $\pi_1$ has to be infeasible
(since the next statement is $\passert(\pfalse)$), i.e., infeasible in
every data model. This is a straight-line program; it has no other
executions and thus the program is correct iff $\pi_1$ is
infeasible. We now check the feasibility of $\pi_1$.

To be infeasible in every data model, $\pi_1$ must, in particular, be
infeasible in the free model of terms. In the free model, equality
corresponds to syntactic equality, and thus to be infeasible in the
free model the values of the variables $x_1$ and $x_2$, namely the
terms $f_{w_\alpha}(\init{x_3})$ and $f_{w_\beta}(\init{x_3})$, must
be syntactically equal at then end of $\pi_2$, which happens iff
$w_\alpha = w_\beta$.

We conclude that an arbitrary program generated by the given grammar
is correct iff $w_\alpha = w_\beta$. Further, any program from the
grammar corresponds to a putative solution to the given PCP instance
$P$, namely the number $N$ and the indices $i_j$ for $1 \leq j \leq N$
such that $w_\alpha = w_\beta$, i.e.,

\[
\alpha_{i_1} \cdot \alpha_{i_2} \cdot \ldots \cdot \alpha_{i_N}
=
\beta_{i_1} \cdot \beta_{i_2} \cdot \ldots \cdot \beta_{i_N}
\]
We conclude that there exists a correct program in the grammar iff
there exists a solution to the PCP instance. Since the original
instance $P$ was arbitrary, this yields undecidability of
uninterpreted program synthesis over schema $\schema_\slp$.\qed

\newpage
\section{Continued from~\secref{coherent-synth}}
\applabel{coh-synth-app}

\subsection{Tree Automata Preliminaries}
\applabel{app-tree-prelim}

\subsubsection{Binary Trees.}
We consider binary trees here.  Let us fix a tree alphabet
$\Gamma = \bigcup_{i=0}^2 \Gamma_i$, which is a finite set of symbols
annotated with arities --- symbols in $\Gamma_i$ have arity $i$ and
$\Gamma_i \cap \Gamma_j = \emptyset$ when $i \neq j$.  Formally, a
finite tree $T$ over $\Gamma$ is a pair $(S, \lab)$, where
$S \subseteq \set{\ldir,\rdir}^*$ is a finite set of \emph{nodes} in
the tree; S is prefix closed, $\epsilon \in S$, and for every string
$\rho{\cdot} \rdir \in S$ we also have $\rho{\cdot}\ldir \in S$.  The
labeling function $\lab : S \to \Gamma$ maps each leaf node $n \in S$
(i.e., those nodes for which there is no node $n'$ which is a suffix
of $n$) to $\Gamma_0$, each node with exactly one child (i.e.,
$n{\cdot}\ldir\in S$ but $n{\cdot}\rdir\not\in S$) to $\Gamma_1$, and
the remaining nodes to $\Gamma_2$.  The node corresponding to
$\epsilon$ is called the \emph{root} node.  The left and right
children of a node $n \in S$ are the nodes $n_1 = n{\cdot}\ldir$ and
$n_2 = n{\cdot}\rdir$ if they exist, in which case $n$ is the parent
of $n_1$ and $n_2$.
For a node $n$ different from the root $\epsilon$, we use
$n{\cdot}\udir$ to denote the parent of $n$.

\subsubsection{Nondeterministic Top-Down Tree Automata.}
A nondeterministic top-down tree automaton over a tree alphabet
$\Gamma = \bigcup_{i=0}^2 \Gamma_i$ is a tuple
$A = (Q, I, \delta_0, \delta_1, \delta_2)$, where $Q$ is a finite set
of states, $I \subseteq Q$ is a set of initial states, and
$\delta_0 \subseteq Q \times \Gamma_0$,
$\delta_1 : Q \times \Gamma_1 \to 2^Q$, and
$\delta_2 : Q \times \Gamma_2 \to 2^{Q \times Q}$ are transitions for
each kind of letter.
For a finite tree $T = (S, \lab)$, a run of $A$ on $T$ is a tree
$\rho = (S, \mu)$ labeled with states of $A$ (i.e., $\mu : S \to Q$)
such that $\mu(\epsilon) \in I$ and for every non-leaf node $n \in S$,
we have $\mu(n{\cdot}\ldir) \in \delta_1(\mu(n), \lab(n))$ if $n$ has
only one child, and
$(\mu(n{\cdot}\ldir), \mu(n{\cdot}\rdir)) \in \delta_2(\mu(n),
\lab(n))$ otherwise.  Further, $\rho$ is accepting if for all leaf
nodes $n \in S$, we have $(\mu(n), \lab(n)) \in \delta_0$ and $T$ is
accepted by $A$ if there is an accepting run of $A$ on $T$.  The
language $L(A)$ of the top-down tree automaton is the set of all trees
it accepts.

We note that checking emptiness of the language of a nondeterministic
top-down tree automaton is decidable in \complexitygrammar~time in the
size of the automaton.  Further, given two tree automata $A_1$ and
$A_2$ over the same tree alphabet, we can construct another tree
automaton $A$ such that $L(A) = L(A_1) \cap L(A_2)$, in time
$O(|A_1| \cdot |A_2|)$.  We refer the reader to~\cite{tata2007} for
details of these standard results.

\subsubsection{Two-Way Alternating Tree Automata.}
In what follows we omit the formal, standard, definitions for labeled
trees. \applabel{2wata-def} We will denote by $\Bb^+(U)$ the set of
all positive Boolean formulae over a set $U$.  That is, $\Bb^+(U)$ is
the smallest set such that
$\set{\texttt{true}, \texttt{false}}\cup U \subseteq \Bb^+(U)$
and for every $\varphi_1, \varphi_2 \in \Bb^+(U)$, we have
$\set{\varphi_1\lor\varphi_2, \varphi_1\land\varphi_2} \subseteq \Bb^+(U)$.
For a (possibly empty) set $U' \subseteq U$ and a formula $\varphi \in \Bb^+(U)$,
we say $U' \models \varphi$ if $\varphi$ evaluates to true
by setting each of the elements in $U'$ to true and
the remaining elements of $U$ to false.

A two-way alternating tree automaton is a tuple
$A = (Q, I, \delta_0, \delta_1, \delta_2)$, where $Q$ is a finite
set of states and $I \subseteq Q$ is the set of initial states.
The functions
$\delta_0, \delta_1$, and $\delta_2$ give, respectively, the
transitions for the leaf nodes, nodes with one child, and nodes with
two children:
\begin{itemize}
	\item $\delta_0 : Q \times \set{\ddir} \times \Gamma_0 \to \Bb^+(Q \times \set{\udir})$
	\item $\delta_1 : Q \times \set{\ddir, \uldir} \times \Gamma_1 \to \Bb^+(Q \times \set{\udir, \ldir})$
	\item $\delta_2 : Q \times \set{\ddir, \uldir, \urdir} \times \Gamma_2 \to \Bb^+(Q \times \set{\udir, \ldir, \rdir})$
\end{itemize}
We explain the key differences between two-way alternating tree
automata and nondeterministic top-down tree automata.  First, unlike
in a top-down tree automaton, where control always moves to children
nodes ($\set{\ldir, \rdir}$), here the control can also move up to the
parent node ($\set{\ldir, \rdir, \udir}$). Second, the input to the
transition function $\delta_i$ is a triple $(q, m, a)$. Here, $q$ and
$a$ are the current state and the label of the current node, as in
top-down tree automata. However, in addition, the transitions depend
on the \emph{last move} $m$. For example, if the automaton moves from
a parent node $n$ to a child node $n\cdot{d}$ (with
$d \in \set{\ldir, \rdir}$), then the last move of the automaton would
be $m=\ddir$, denoting a `downward' move.  Similarly, if the automaton
moves from a child node $n{\cdot}\ldir$ to the parent node $n$, then
the last move would be $m=\uldir$ denoting `upward' move from the
`left' child. A third difference comes from alternation, which is more
general than nondeterminism--- the control can move to any set of
states that satisfy the Boolean formula given by the transition
function. These differences are formalized in the definition of a run
for such an automaton, which we describe next.

\newcommand{\trun}{\text{run}}
A \emph{run} of a two-way alternating tree automaton $A$ over a finite,
binary labeled tree $T = (S, \lab)$ is a (possibly infinite) directed
acyclic rooted labeled tree\footnote{A rooted directed graph is a tree
  if there is a unique directed path from the root to every other
  vertex in the graph.} $G_\trun = (V_\trun, \lab_\trun, E_\trun)$,
rooted at a designated vertex $r \in V_\trun$.  The labeling function
is of type
$\lab_\trun : V_\trun \to S\times Q \times \set{\ddir,\uldir,
  \urdir}$, and the graph $G_\trun$ satisfies the following
conditions:
\begin{enumerate}[label=(\alph*)]
\item The root $r \in V_\trun$ has label
  $\lab_\trun(r) = (\epsilon, q, \ddir)$, where $q \in I$.

\item
For every node $v \in V_\trun$ with $\lab_\trun(v) = (n, q, m)$ and every child node $v'$ of $v$ (i.e., $(v, v') \in E_\trun$) with $\lab_\trun(v') = (n', q', m')$, we have
\begin{itemize}
  \item
  if $m' = \ddir$, then $n' = n{\cdot}\ldir$ or $n' = n{\cdot}\rdir$,
  \item
  if $m' = \uldir$, then $n = n'{\cdot}\ldir$, and
  \item
  if $m' = \urdir$, then $n = n'{\cdot}\rdir$.
\end{itemize}
Here, we have glossed over the case when $n$ is the root.  In this
case, we might have $n' = n = r$ and $m' = \uldir$.
\item
For every node $v \in V_\trun$ with $\lab_\trun(v) = (n, q, m)$, the set $C_v = \setpred{v'}{(v, v') \in E_\trun}$ of children of $v$ is such that
    $\setpred{(q', d')}{\exists v'\in C_v, \lab_\trun(v') = (n', q', m'), d' = \text{dir}(n, n')} \models \delta_i(q, m, \lab(n))$,
    where $\text{dir}(n, n')$ is $\ldir$, $\rdir$ or $\udir$ if $n'$
    is respectively the left child, right child or the parent of $n$
    in the input tree $T$.  (Here again, if $n = n' = r$, then we say
    $\text{dir}(n, n') = \udir$.)
\end{enumerate}
A run is accepting if every node in $V_\trun$ has $\geq 1$ child,
i.e. all paths starting from the root are infinite.  
An automaton $A$ \emph{accepts} $T$ if there is an accepting run of
$A$ on $T$.

\newcommand{\game}{\text{game}} An alternative presentation of the
acceptance condition of a two-way alternating tree automaton can be
made in terms of 2-player games.  Such a game is played over a
directed graph $G_\game = (V_\game, E_\game)$.  The set of vertices is
$V_\game = V^0_\game \cup V^1_\game$, where
$V^0 = S \times Q \times \set{\ddir,\uldir, \urdir}$ and
$V^1 = \mathcal{P}(Q \times \set{\udir,\ldir, \rdir})$.  The set of
edges is $E_\game = E^{01}_\game \cup E^{10}_\game$, where
$E^{ij}_\game \subseteq V^i_\game \times V^j_\game$ (where
$i\neq j \in \set{0,1}$).  Let us describe the first set of edges
$E^{01}_\game$.  We have $(u, v) \in E^{01}_\game$ iff
$v \models \delta_i(q, m, \lab(n))$, where $u = (n, q, m)$.  We have
$(v, u) \in E^{10}_\game$ iff $\exists (q', d') \in v$ such that
$u = (n', q', m')$, where
\begin{itemize}
  \item if $d' = \udir$ and $n = r$, then $n' = n$ and $m' = \uldir$
  \item if $d' = \udir$ and $n\neq r$, then either $n = n'{\cdot}\ldir$ and $m' = \uldir$ or
  $n = n'{\cdot}\rdir$ and $m' = \urdir$
  \item if $d' = \ldir$, then $n' = n{\cdot}\ldir$ and $m' = \ddir$
  \item if $d' = \rdir$, then $n' = n{\cdot}\rdir$ and $m' = \ddir$
\end{itemize}

We now describe a \emph{play} on this game graph.  A play starts in
some vertex $u_0 = (\epsilon, q_0, \ddir) \in V^0_\game$ with
$q_0 \in I$.  First, Player-$0$ picks a neighbor $v$ of $u_0$.  Next,
Player-$1$ picks a neighbor of $v$, from which Player-$0$ then
continues the play. This process either goes on indefinitely or ends
because either player reaches a node with no outgoing edges.  In the
former case Player-$0$ wins the play, and in the latter case
Player-$1$ wins the play.  For a given game graph, there can be many
different plays as defined above, some in which Player-$0$ wins and
others in which Player-$1$ wins.  The input tree $T$ is accepted by
the 2-way automaton if there is a play in which Player-$0$ wins.
Readers may observe that the two notions of acceptance are equivalent.

We can convert two-way alternating tree automata to equivalent
nondeterministic top-down tree automata with at most exponential
increase in states.
\begin{lemma}[\cite{vardikupferman2000,Vardi98}]
  \lemlabel{2way-to-nondet-lemma} Given a two-way alternating tree
  automaton $A$, we can construct a nondeterministic top-down tree
  automaton $A'$ of size $O(2^{\emph{\textsf{poly}}(|A|)})$ in time
  $O(2^{\emph{\textsf{poly}}(|A|)})$ such that $L(A') = L(A)$.
\end{lemma}
Here, $|A|$ denotes the size of the description of $A$.
\appref{2way-to-top-down-proof} presents a construction for the above
result adapted to the simpler setting of this paper.



\subsection{Proof of~\lemref{2way-to-nondet-lemma}}
\applabel{2way-to-top-down-proof}

The ingredients for the proof of~\lemref{2way-to-nondet-lemma} are as
follows.
\begin{enumerate}
\item We first characterize when an input tree $T$ is accepted by a
  given 2-way automaton $A$, in terms of what we call \emph{strategy
    annotations}.  As we saw, a tree $T$ induces a $2$-player game
  such that $T$ is accepted iff Player-$0$ wins.  This game is a
  finite \emph{parity} game, and is thus \emph{determined} and admits
  a \emph{memoryless
    strategy}~\cite{emerson-jutla,Mostowski1991GamesWF}. That is, for
  each node $v$ of the game graph, one of the players wins if the game
  starts in $v$.  Further, when we fix the initial node
  $v_\text{init}$ of the game graph, the winner, Player-$i$, has a
  memoryless strategy. A memoryless strategy for Player-$i$ is a
  mapping $V^i_\game \to V^{1-i}_\game$ which maps each vertex
  $v \in V^i_\game$ of Player-$i$ to the next move (neighboring
  vertex) which Player-$i$ should choose each time they visit $v$.
\item With the above observation, it follows that we can annotate
  nodes in $T$ with strategy information for the corresponding game
  graph.  If we can annotate the nodes of $T$ with a winning strategy
  for Player-$0$, then the tree is accepted. Otherwise, $T$ is
  rejected.
\item We build a nondeterministic top-down tree automata $A'$ which
  reads an input tree $T$, nondeterministically decorates its nodes
  with strategy annotations, and determines if the annotation
  corresponds to a winning strategy for Player-$0$.  
\end{enumerate}

\subsubsection{Strategy Annotations.}
Given a labeled binary tree $T$, the game graph $G_\game$ induced by
$T$ is played over vertices $V^0_\game \cup V^1_\game$.  A memoryless
strategy for Player-$0$ in such a game is a mapping
$\beta : V^0_\game \to V^1_\game$, or equivalently, a mapping of the
form
$\beta: S \times Q \times \set{\uldir, \urdir, \ddir} \to \Pp(Q \times
\set{\udir, \ldir, \rdir})$.  We can curry this representation to get
a strategy for each node in $T$, specifically, a partial function of
the form
$\strategy: Q \times \set{\uldir, \urdir, \ddir} \hookrightarrow \Pp(Q
\times \set{\udir, \ldir, \rdir})$.  We denote by $\strategies_A$ the
set of such node-specific strategies for $A$. Observe the size of this
set is $O(2^{O(|Q|^2)})$.  Given a tree $T = (S, \lab)$, a strategy
annotation $\stannot$ maps each node of $T$ to some strategy, i.e.,
$\stannot: S \to \strategies_A$ such that it satisfies the transitions
of $A$: for every $n \in S$ and for every
$q \in Q, m \in \set{\uldir, \urdir, \ddir}$ such that
$\stannot(n)(q, m)$ is defined, we have that
$\stannot(n)(q, m) \models \delta_i(q, m, \lab(n))$ (where $i$ is
$0, 1$, or $2$ depending upon the label $\lab(n)$, and $m$ is also
appropriately chosen depending upon the arity of $n$).  We remark that
if for some node $n$ of arity $i$, state $q$ and direction $m$, if
$\delta_i(q, m, \lab(n)) = \texttt{false}$, then $\stannot(n)(q, m)$
cannot be defined for any strategy annotation $\stannot$.

\newcommand{\dom}{\textsf{dom}} For every tree node $n \in S$, let
$\dom(\str) = \setpred{(q, m)}{\str(q, m) \text{ is defined}}$.  A
strategy annotation $\stannot$ over $T$ is \emph{consistent} if for
every $n \in S$ and every $(q, m) \in \dom({\stannot(n)})$, for all
$(q', d') \in \stannot(n)(q, m)$ we have
$(q', m') \in \dom({\stannot(n')})$, where $n' = n{\cdot}d'$ and
$m' = \uldir$ if $n = n'{\cdot}\ldir$, $m' = \urdir$ if
$n = n'{\cdot}\rdir$ and $m' = \ddir$ if $n' = n{\cdot}\ldir$ or
$n' = n{\cdot}\rdir$.  We formalize below the equivalence between the
existence of a consistent strategy annotation of $T$ and its
acceptance by $A$.  The proof follows directly from the fact that the
game induced by $T$ has a memoryless winning strategy for Player-$0$
iff Player-$0$ wins (or equivalently $T$ is accepted).

\begin{proposition}
  \proplabel{accepting-stannot} Let $A$ be a two-way alternating tree
  automaton and let $T$ be an input tree. Then $T$ is accepted by $A$
  iff there is a  consistent strategy annotation $\stannot$ of $T$.
\end{proposition}

\subsubsection{Construction of an Equivalent Top-Down Automaton.}

We now describe the nondeterministic top-down tree automaton $A'$ that
accepts the same language as the given two-way alternating tree
automaton $A$. At a high level, the construction is made possible
by~\propref{accepting-stannot}, which suggests that we \emph{guess} a
strategy annotation for the input tree in one shot. The challenge,
however, is to verify in a top-down manner that the guessed annotation
is consistent.  Observe that the domain $\dom(\stannot(n))$ for any
$\stannot$ and any $n$ is finite (a subset of
$Q\times \{\uldir,\urdir,\ddir\}$).  Thus we can guess these subsets
(call them \emph{bags}) for each tree node.  Verifying that the
guessed strategy annotation is consistent then reduces to checking if
from every pair $(q, m)$ in the bag of a node $n$ and each possible
$(q', d') \in \stannot(n)(q,m)$, we have that the corresponding
$(q', m')$ (where $m'$ defined using $n$ and $d'$) is in the bag of
the next node $n'$.
We formalize the construction below.

\newcommand{\sink}{\textsf{Sink}} Fix a two-way alternating tree
automaton $A = (Q, I, \delta_0, \delta_1, \delta_2)$.  The top-down
automaton is a tuple
$A' = (Q'\uplus \set{\sink}, I', \delta'_0, \delta'_1, \delta'_2)$.
$\sink$ is a special absorbing state.  The set
$Q' = \Gamma \times \strategies_A$ consists of pairs of labels and
strategies, such that for every $q' = (a, \sigma) \in Q'$ we have:
\begin{enumerate}
  \item $\dom(\str) \neq \emptyset$,
  \item for every $(q, m) \in \dom(\str)$, $m$ is appropriate for the
    arity $i$ of $a$,
\item  for every $(q, m) \in \dom(\str)$,
$\str(q, m) \models \delta_i(q,m,a)$,
and in particular $\delta_i(q,m,a) \neq \texttt{false}$.
\end{enumerate}

We now describe the transitions.  First, for every $i$ and every $a$,
$\delta'_i(\sink, a) = \set{\sink}$.  Second, for every $i$ and every
$b \neq a$, $\delta'_i((a, \str), b) = \set{\sink}$.  All other
transitions are described below.

\begin{description}
\item[Transitions on nodes with 1 child.]  We have
  $(a', \str') \in \delta'_1((a, \str), a)$ iff the following hold.
\begin{enumerate}[label=(\alph*)]
\item For every $(q, m) \in \dom(\str)$ and
  $(q', \ldir) \in \str(q, m)$, $(q', \ddir) \in \dom(\str')$.
\item For every $(q, m) \in \dom(\str')$ and
  $(q', \udir) \in \str'(q, m)$, $(q', \uldir) \in \dom(\str)$.
\end{enumerate}

\item[Transitions on nodes with 2 children.]  We have
  $((a'_1, \str'_1), (a'_2, \str'_2)) \in \delta'_2((a, \str), a)$ iff
  the following hold.
\begin{enumerate}[label=(\alph*)]
    \item For every $(q, m) \in \dom(\str)$ and $(q', \ldir) \in \str(q, m)$, $(q', \ddir) \in \dom(\str'_1)$.
    \item For every $(q, m) \in \dom(\str'_1)$ and $(q', \udir) \in \str'_1(q, m)$, $(q', \uldir) \in \dom(\str)$.
    \item For every $(q, m) \in \dom(\str)$ and $(q', \rdir) \in \str(q, m)$, $(q', \ddir) \in \dom(\str'_2)$.
    \item For every $(q, m) \in \dom(\str'_2)$ and $(q', \udir) \in \str'_2(q, m)$, $(q', \urdir) \in \dom(\str)$.
\end{enumerate}

\item[Transition relation for leaf nodes.]  Finally,
  $\delta'_0 \subseteq Q' \times \Gamma_0$ is such that\newline
  $((a, \str), b) \in \delta'_0$ iff $a = b$.
\end{description}

Observe that the transitions are designed to mimic the notion of
consistency for strategy annotations.  Whenever a nondeterministic
transition fails to ensure consistency, the resulting state is
$\sink$.  Further, the transitions also guess the label of the next
tree node, and a failure to correctly guess this also leads to the
state $\sink$.  Finally, given the definition of $\delta'_0$ and the
fact that $\sink$ is absorbing, a tree is accepted by $A'$ iff there
is a consistent strategy annotation for the tree.  The correctness of
the construction is straigtforward.  The complexity result follows
from the observation that the number of states of $A'$ is
$O(|\strategies_A|) = O(2^{\text{poly}(|A|)})$ and the size of
transitions for $A'$ is
$O(\text{poly}(|\strategies_A|)) = O(2^{\text{poly}(|A|)})$.
\vspace{0.1in}

\begin{replemma}{lem:2way-to-nondet-lemma}
  Given a two-way alternating tree automaton $A$, we can construct a
  nondeterministic top-down tree automaton $A'$ of size
  $O(2^{\emph{\text{poly}}(|A|)})$ in time $O(2^{\emph{\text{poly}}(|A|)})$ such
  that $L(A') = L(A)$.
\end{replemma}

\subsection{Top-Down Tree Automata for Program Trees}
\applabel{grammar-to-aut}
\subsubsection{Grammar to Tree Automaton.}
The next task is to represent the set of programs generated by an
input grammar $\gram$ as a regular set of program trees. We construct
a top-down tree automaton $\treeaut{A}_\gram$ which accepts precisely
the set of trees that correspond to the programs generated by
$\gram = (\Delta_0\uplus\Delta_1\uplus\Delta_2, St, NT, R)$.  We
require that $\gram$ conforms to the schema $\schema$ discussed
in~\secref{schema}.

We now define the components of
$\treeaut{A}_\gram = (Q^\gram, I^\gram, \delta_0^\gram,
\delta_1^\gram, \delta_2^\gram)$.  The states of $\treeaut{A}_\gram$
correspond to the nonterminals of $\gram$ plus a special starting
state $q_0$. That is, $Q^\gram = \set{q_0} \uplus
NT$,
and $I^\gram = \set{q_0}$.
The transitions are defined as follows.
\begin{align*}
  \delta^\gram_0 &= \setpred{(P, a)}{\dblqt{P \produces \, a} \in R, a \in \Delta_0} \\
  \delta^\gram_1(q, a) &=
\begin{cases}
\set{St} & \text{ if } q = q_0, a = \dblqt{\proot} \\
\setpred{P_1}{\dblqt{P \produces \, \pwhile \, (x \sim y) \,\, P_1} \in R} & \text{ if } q = P, a = \dblqt{\pwhile(x \sim y)} \\
\emptyset & \text{ otherwise }
\end{cases} \\
\delta^\gram_2(P, a)
&=
\begin{cases}
\setpred{
(P_1, P_2)
}{
	\dblqt{P \produces \, \pif \, (x \sim y) \, \pthen \, P_1 \, \pelse \, P_2} \in R
}
& \text{ if } a = \dblqt{\pite}(x \sim y) \\
\setpred{
(P_1, P_2)
}{
\dblqt{P \produces \, P_1 \, ; \, P_2} \in R
}
& \text{ if } a = \dblqt{\pseq} \\
\end{cases}
\end{align*}

The following lemma states that the language of the tree automaton
constructed above accurately represents programs from $\gram$.  The
notation $\prog(T)$ refers to the word obtained by collecting labels
during an in-order traversal of tree $T$.
\begin{lemma}
  Let $\gram$ be a grammar conforming to the schema $\schema$ and let
  $\treeaut{\Aa}_\gram$ be the tree automaton constructed above. Then,
  we have
  $L(\gram) = \setpred{\emph{\prog(\emph{T})}}{T \in
    L(\treeaut{A}_\gram)}$.  Further, $\treeaut{A}_\gram$ has size
  $O(|\gram|)$ and can be constructed in time $O(|\gram|)$.\qed
\end{lemma}

\newpage

\section{Proof of Lower Bound for Program Synthesis}

\applabel{proof-of-lower-bound}

Let $M = (Q, \Delta, \delta, q_0, g)$ be a single-tape alternating
Turing machine (ATM) with exponential space bound, where $Q$ and
$\Delta$ are finite sets of states and tape symbols, respectively. The
transition function has the form
$\delta : (Q \times \Delta) \rightarrow \Pp(Q \times \Delta \times
\{L, R\})$. Without loss of generality, we will assume that either
there exist exactly two transitions (referred to as \code{0} and
\code{1}) for any particular configuration or none at all. The initial
state is $q_0 \in Q$ and
$g : Q \rightarrow \{acc, rej, \wedge, \vee\}$ maps states to their
type. It will be convenient to represent machine configurations as
sequences of symbols. For a given machine $M$, this can be allowed by
working with a modified machine $M'$ whose alphabet
$\Gamma = \Delta \cup (Q \times \Delta)$ contains the original symbols
from $\Delta$ as well as composite symbols from $(Q \times \Delta)$ to
encode both the machine head position and machine state. For example,
for state $q \in Q$ and symbol $t \in \Delta$, the composite symbol
$(q, t) \in \Gamma$ encodes the tape head reading regular symbol $t$
with the machine in state $q$. The transition function $\delta$ can
easily be modified to account for this representational change, and we
omit the details. A \emph{universal} (resp. \emph{existential})
configuration is a sequence of symbols over $\Gamma$ containing
exactly one composite symbol $(q, t)$, with
$g(q) = \wedge \, \, (\text{resp. } g(q) = \vee)$. 
The set of \emph{accepting configurations} is the smallest set $S$
that contains (a) all configurations whose state $q$\textbf{} has
$g(q) = acc$, (b) all universal configurations such that \emph{every}
configuration reachable within one transition belongs to $S$ and (c)
all existential configurations such that there is \emph{some}
configuration in $S$ reachable within one transition. An ATM $M$
accepts an input $w$ if the initial configuration is contained in
$S$. In what follows, we assume without loss of generality that
existential configurations are immediately followed by universal
configurations under any transition, and vice versa. Further, we
assume the initial configuration is existential, and any configuration
whose state $q$ has $g(q)\in\{acc,rej\}$ is halting, i.e. there is no
transition.

Our representation of configurations as sequences of tape symbols
allows us to work with a modified transition relation $\delta_W$ that
is lifted to configuration windows.  A configuration window is a
triple of tape symbols. After endowing an ATM with compositite
symbols, as mentioned above, the information in $\delta$ can be easily
represented by $\delta_W \subseteq (\Gamma^3 \times \Gamma)$, which
relates triples of tape symbols to symbols that the middle cell can
legally transition to according to $\delta$.  For convenience, in the
reduction grammar we will overload $\delta_W$, writing
$\delta_W(0, t_i, t_j, t_k)$ to denote the tape symbol for the cell
containing $t_j$ (with $t_i$ and $t_k$ to the left and right) after
the machine takes the \code{0} transition.  Similarly,
$\delta_W(1, t_i, t_j, t_k)$ will denote the tape symbol for the
\code{1} transition. If the window $(t_i, t_j, t_k)$ is ill-formed
(for example, it may contain more than one composite machine state
symbol) we say $\delta_W(0, t_i, t_j, t_k)$ and
$\delta_W(1, t_i, t_j, t_k)$ are undefined. We omit details for
handling corner cases in which the machine reads at either edge of the
tape. This can be easily dealt with by assuming special edge-of-tape
symbols. In the forthcoming grammar, we assume $|\Gamma| = k$ and use
$t_1 \ldots t_k$ as constants to model the (extended) alphabet of a
(appropriately modified) machine $M$. Indices for such variables are
sometimes used to indicate the particular symbol they contain, e.g.
$t_{blank}$ refers to the unique blank symbol.

\newpage
\subsection{Reduction}
\label{sec:reduction}

\begin{minipage}{0.6\textwidth}
  Before presenting the reduction grammar $\gram_{M, w}$, we discuss
  its primary components and summarize the purposes of its program
  variables. Recall that the goal of the reduction is to produce a
  grammar $\gram_{M, w} \in \schema$ that contains a correct program
  exactly when an $\aexpspace$ turing machine $M$ accepts an input
  $w$. Essentially, this amounts to building a grammar whose correct
  programs encode winning strategies for player Eve in a two-player
  game semantics. That is, we want a correct program to exist in
  $\gram_{M, w}$ exactly when there is a configuration tree starting
  from the initial configuration of $M$ on $w$, branching
  appropriately according to $\delta$, and terminating in accepting
  leaf configurations. The grammar encodes the alternation by
  branching on a \code{turn} variable (for players Adam and Eve) and
  choosing the next transition accordingly. For Eve's turn, the
  grammar enforces a transition choice from the synthesizer, whereas
  Adam's turns are read from the uninterpreted function
  \code{choice}. Relying on the intuition that correct programs must
  satisfy assertions in every data model, it is straightforward to see
  that this captures the semantics of ATMs. Deeper in the grammar,
  observe that after a move has been made, there is a mechanism for
  requiring that the next configuration (produced by the synthesizer)
  indeed follows from the selected move according to $\delta$. As
  noted earlier, the full configuration cannot be represented in
  program variables all at once because $M$ may use exponential space
  (and hence a grammar using exponentially many variables could not be
  produced in polynomial time). To solve this, the grammar uses a loop
  to iterate through the full configuration, enforcing that the
  synthesizer produces the configuration contents one cell at a
  time. See \figref{StrategyTree} for an illustration of this idea.
\end{minipage}
\hspace{0.5cm}
\begin{minipage}{0.4\textwidth}
\centering

\begin{figure}[H]
\scalebox{0.75}{
\begin{tikzpicture}
\node (bigc1) at (0,0) [draw=drawioBlue!80!black, circle, minimum width = 4em, thick] {};
\node (t1) at (0,-2.5) [draw=drawioBlue!80!black, text=black, circle, minimum width = 1em] {$t_1$};
\node (one) at (1.5,-2.5)  {\Large $1$};
\node (t3) at (0,-4) [draw=drawioBlue!80!black, text=black, circle, minimum width = 1em, outer sep=4pt] {$t_3$};
\node (two) at (1.5,-4)  {\Large $2$};
\node (t7) at (0,-5.5) [draw=drawioBlue!80!black, text=black, circle, minimum width = 1em, outer sep=4pt] {$t_7$};
\node (one) at (1.5,-5.5)  {\Large $2^\text{poly(m)}$};
\node (sq) at (0,-7.5) [draw=drawioRed!80!black, text=black, rectangle, minimum width = 4em, thick, minimum height = 4em] {};
\node (bigc2) at (-2,-9.5) [draw=drawioBlue!80!black, circle, minimum width = 4em, thick] {};
\node (bigc3) at (2,-9.5) [draw=drawioBlue!80!black, circle, minimum width = 4em, thick] {};
\node (t5) at (-2,-12) [draw=drawioBlue!80!black, text=black, circle, minimum width = 1em] {$t_5$};
\node (t2) at (-2,-13.5) [draw=drawioBlue!80!black, text=black, circle, minimum width = 1em, outer sep=4pt] {$t_2$};
\node (emptyl) at (-2,-14.5) [outer sep=4pt] {};
\node (t16) at (2,-12) [draw=drawioBlue!80!black, text=black, circle, minimum width = 1em] {$t_{16}$};
\node (t8) at (2,-13.5) [draw=drawioBlue!80!black, text=black, circle, minimum width = 1em, outer sep=4pt] {$t_8$};
\node (emptyr) at (2,-14.5) [outer sep=4pt] {};

 \draw (bigc1) edge[-{Latex[length=2mm, width=2mm]}, thick]   node[fill=white, minimum width=0em] {\bf 0} (t1);
 \draw (t1) edge[-{Latex[length=2mm, width=2mm]}, thick] (t3);
 \draw[dotted, ultra thick] (t3) -- (t7);
 \draw (t7) edge[-{Latex[length=2mm, width=2mm]}, thick] (sq);
 \draw (sq) edge[-{Latex[length=2mm, width=2mm]}, thick, bend right] node[fill=white, minimum width=0em] {\bf 0} (bigc2);
 \draw (sq) edge[-{Latex[length=2mm, width=2mm]}, thick, bend left] node[fill=white, minimum width=0em] {\bf 1} (bigc3);
  \draw (bigc2) edge[-{Latex[length=2mm, width=2mm]}, thick]   node[fill=white, minimum width=0em] {\bf 1} (t5);
   \draw (bigc3) edge[-{Latex[length=2mm, width=2mm]}, thick]   node[fill=white, minimum width=0em] {\bf 0} (t16);
 \draw (t5) edge[-{Latex[length=2mm, width=2mm]}, thick] (t2);
  \draw (t16) edge[-{Latex[length=2mm, width=2mm]}, thick] (t8);
 \draw[dotted, ultra thick] (t2) -- (emptyl);
 \draw[dotted, ultra thick] (t8) -- (emptyr);
\end{tikzpicture}
}
\caption{ \figlabel{StrategyTree} A strategy tree for Eve. Large
  circles represent Eve's moves, squares represent Adam's. Small
  circles represent steps in which Eve chooses a tape symbol and Adam
  checks it. Each move requires Eve to output an exponential number of
  tape symbols $t_i$, one after the other. She must be able to do this
  for each of Adam's moves in addition to her own. Such a strategy, if
  Adam indeed agrees with the tape symbols, witnesses an accepting
  computation tree for the ATM. }
\end{figure}
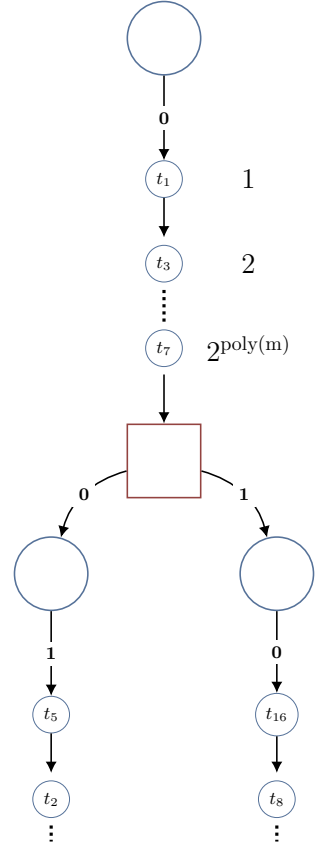
\end{minipage}

Further, since the full configuration contents cannot be stored at
once, the correctness check must distribute the work across all data
models. For an input $w$ with $m = |w|$, the grammar utilizes
$\textit{n = poly(m)}$ index variables $\code{s}_1 \ldots \code{s}_n$
to point into the configuration. For any given data model, the index
points at a single tape cell. For this single cell, the grammar
enforces that all transitions are correct. Since uninterpreted
programs must be correct in all data models, it follows that a correct
program from the target grammar will witness the correctness of
transitions for all tape cells. Finally, the grammar enforces that any
leaf configurations are accepting. Table~\ref{tab:SummaryOfPurposes}
summarizes the purposes of all variables in the grammar. We denote
vectors of variables with boldface, e.g.
$\code{s}_1 \ldots \code{s}_n$ by \code{\codekey{s}}.

\vspace{0.1in}
\begin{table}
\begin{tabular}{rl}
  \hline \\
  $\code{w}_1 \ldots \code{w}_m$:& store input tape contents \\
  $\code{t}_1 \ldots \code{t}_k$:& constants to represent tape symbols \\
  $\code{s}_1 \ldots \code{s}_n$:& index holding location of the cell being
                     checked \\
  $\code{s}_1' \ldots \code{s}_n'$:& index holding binary predecessor to
                       $\code{s}_1 \ldots \code{s}_n$ \\
  $\code{s}_1'' \ldots \code{s}_n''$:& index holding binary successor to
                         $\code{s}_1 \ldots \code{s}_n$ \\
  $\code{b}_1 \ldots \code{b}_{n+1}$:& index for pointing at current tape
                         cell during iteration \\
  \code{0}, \code{1}:& constants used for indices as well as move choices \\
  \code{cell}:& holds the putative current tape symbol \\
  \code{turn}:& holds \code{0} or \code{1}, represents which player has the next move \\
  \code{dir}:& holds \code{0} or \code{1}, the most recent move choice \\
  $\code{c}_1, \code{c}_2, \code{c}_3$:& the previous contents of the configuration
                    window for the cell being checked \\
  $\code{c}_1', \code{c}_2', \code{c}_3'$:& the next contents of the configuration
                       window for the cell being checked \\
  \code{choice}:& uninterpreted function that models Adam's moves,\\
                   &restricted to \code{0} and \code{1} with an assume statement \\
  \code{next}:& uninterpreted function that captures the
                time-dependence of Adam's decisions \\
  \\
  \hline
  \vspace{0.25in}
\end{tabular}
  \caption{\label{tab:SummaryOfPurposes}Summary of purposes for
    grammar variables. }
\end{table}

\figref{Grammar1} and~\figref{Grammar2} present the full grammar
$\gram_{M, w}$. The reader may note that $\gram_{M, w}$ does not
appear to conform to grammar schema $\schema$. It is not hard to see
that in fact $\gram_{M, w}$ can be factored appropriately by
introducing a polynomial number of new nonterminal symbols to produce
the components involving fixed instruction sequences. Before arguing
that our reduction is correct, we pause to mention a few
presentation-related details and to explain the important rules for
$\gram_{M, w}$.

To simplify the presentation of the grammar and promote its
interesting rules, several simplifications and omissions were
made. First, we omit $\pelse$ branches whenever they include only a
$\pskip$ statement. Second, several of the conditions in $\pif$ and
$\passume$ statements consist of boolean combinations of equality and
disequality. These can be translated into semantically equivalent
statements using sequences of nested $\pif$ statements. For each
condition in $\gram_{M, w}$, the translated code is of polynomial
size. This crucially relies on the fact that each condition is already
expressed in conjunctive normal form. We omit the precise translation,
which is straightforward. A few places in the grammar use variables
without subscripts to denote the bitvector representation of a number
$n$, e.g. \code{\codekey{n}}, and subscripted variables to access the
components. Equalities over bitvectors (e.g.
$\pif \,\code{(\codekey{b} = \codekey{n})}$) ultimately are handled
using a conjunction of equalities on each bit. Overflow and underflow
in the binary operation rules are not addressed, but could easily be
fixed by adding conditional statements and keeping a few variables as
flags to signal such events. Finally, when two bitvectors of disequal
length are compared (e.g. \code{\codekey{b}} and \code{\codekey{s}} in
\emph{Check}), the bit-wise comparison of the least significant bits
is intended.

We now trace the important parts of the grammar, and encourage the
reader to follow along in~\figref{Grammar1}. The \emph{Move} rule
extends the strategy tree by one move or, alternatively, finishes it
in \emph{Base} by asserting that configurations are accepting. For
extending the tree, the grammar checks which player's turn is next
with a conditional statement. On one branch the synthesizer is allowed
to choose a transition to make, and on the other the transition is
determined by reading from \code{choice}. After determining the next
move in each branch, the \emph{Generate} rule simulates and checks
it. The synthesizer iteratively produces the contents of each tape
cell inside a $\pwhile$ loop. It is allowed to branch on the index
variables that determine the current iteration (and \emph{not} the
secret index \code{\codekey{s}}) in order to decide which symbol to
produce (see \emph{ProduceCell}.)  The grammar checks correctness of a
transition only for the particular cell it happens to be tracking (see
\emph{Check}). After this, the process is repeated with another
\emph{Move}. Recall that in general, a strategy for Eve requires more
memory than can be explicitly allocated in program variables. The
grammar provides for this memory by placing \emph{Generate} rules
under each branch of the \emph{Move} rule. This has the effect of
using the program counter as memory, as mentioned earlier.

Our grammar can be produced in time polynomial in $m = |w|$. There are
polynomially many grammar rules, and each is clearly of size
polynomial in $m$. The key ideas that make this possible are a bounded
loop to produce tape contents and distributing the check for
transition correctness across data models. Finally, memory is provided
to the synthesizer by choosing a grammar that allows the program
structure to grow as a tree, with branches encoding the move history.


\begin{figure}[H]
\vrule \hspace{0.1cm}
\begin{minipage}[t]{.48\textwidth}
\begin{lstlisting}[basicstyle=\scriptsize,mathescape=true,escapechar=\#]
  $\production{bit}$   $\produces$ #\code{0}# $\mid \ldots \mid$ #\code{1}#
  $\production{tape}$ $\produces$ #\code{t$_1$}# $\mid \ldots \mid$ #\code{t$_k$}#
  $\production{idx}$   $\produces$ #\code{b$_1$}# $\mid \ldots \mid$ #\code{b$_{n+1}$}#

  $\production{S} \produces \production{Prelude} \code{;} \production{Move}$

  $\production{Move} \produces$
             $\mid \production{Base}$
             $\mid \pif \code{(turn = 1)}$
                     $\code{dir} \passign \production{bit} \code{;}$
                     $\code{turn} \passign \code{0;}$
                     $\production{Generate}$
               $\pelse$
                     $\code{dir} \passign \code{choice(univ);}$
                     $\passume \code{(dir = 0} \vee \code{dir = 1);}$
                     $\code{univ} \passign \code{next(univ);}$
                     $\code{turn} \passign \code{1;}$
                     $\pif \code{(dir = 0)} \production{Generate}$
                     $\pelse$         $\production{Generate}$

  $\production{Generate} \produces$
             $\code{\codekey{b}} \passign \code{\codekey{0};}$
             $\pwhile$ #\code{(b$_{n+1}$ $\neq$ 1)}#
                     $\production{ProduceCell} \code{;}$
                     $\production{Check} \code{;}$
                     $\production{Increment} \code{;}$
             $\production{Move}$

\end{lstlisting}
\end{minipage}  
\vrule 
\begin{minipage}[t]{.48\textwidth}
\begin{lstlisting}[basicstyle=\scriptsize,mathescape=true,escapechar=\#]
  $\production{Base} \produces$
               $\pif \code{(} \bigvee_{\code{t} \in \Gamma \wedge \,comp(\code{t})}$ #\code{c$_2$ = t)}#
                   $\passert$ #\code{($\bigvee_{\code{t}' \in \Gamma \wedge \,acc\code{(t}'\code{)}}$ c$_2$ = t$'$)}#

  $\production{ProduceCell} \produces$
             $\mid \pif \code{(} \production{idx} \code{ = } \production{bit} \code{)}$
                     $\production{ProduceCell}$
               $\pelse$
                     $\production{ProduceCell}$
             $\mid \code{cell} \passign \production{tape}$

  $\production{Check} \produces$
             $\pif \code{(\codekey{b} = \codekey{s}}'\code{)}$
                     #\code{c$_1'$}# $\passign \code{cell}$
             $\pif \code{(\codekey{b} = \codekey{s})}$
                     #\code{c$_2'$}# $\passign \code{cell}$
                     $\pif$ #\code{(c$_1$ = t$_i$ $\wedge$ c$_2$ = t$_j$}#
                                     #\code{$\wedge$ c$_3$ = t$_l$)}#
                             $\pif \code{(dir = 0)}$
                                     $\passert$ #\code{(cond$_0$)}#
                             $\pelse$
                                     $\passert$ #\code{(cond$_1$)}#
                     $\pelse\,\, \pif$ #\code{(c$_1$ = t$_{i'}$ $\wedge$ c$_2$ = t$_{j'}$}#
                                           #\code{$\wedge$ c$_3$ = t$_{l'}$)}#
                           #\code{$\ldots$}#
             $\pif \code{(\codekey{b} = \codekey{s}}''\code{)}$
                     #\code{c$_3'$}# $\passign \code{cell}$
             #\code{c$_1 \passign$ c$_1'$; c$_2 \passign$ c$_2'$; c$_3 \passign$ c$_3'$}#
\end{lstlisting}
\end{minipage}
\vrule
\caption{\figlabel{Grammar1} These rules impose the strategy tree
  structure and check correctness of each move. Note that \emph{Check}
  elides the full branching on possible windows and uses a shorthand
  condition \,\passert\code{(cond$_x$)}\, to denote
  \,\passert\code{(c$_2$ = $\delta_W$(x, t$_i$, t$_j$, t$_k$))}\,
  whenever \,\code{$\delta_W$(x, t$_i$, t$_j$, t$_k$)}\, is defined,
  and to denote \,\passert(\pfalse)\, otherwise. The shorthand
  notation $comp(\code{t})$ holds for any composite symbol
  $\code{t} = (x, q) \in \Gamma$, and $acc(\code{t})$ holds for any
  composite symbol $\code{t} = (x, q)$ with $g(q) =
  acc$. See~\figref{Grammar2} for \emph{Prelude} and
  \emph{Increment}. }
\end{figure}


\begin{figure}[H]
\vrule \hspace{0.1cm}
\begin{minipage}[t]{.48\textwidth}
\begin{lstlisting}[basicstyle=\scriptsize,mathescape=true,escapechar=\#]
  $\production{Prelude} \produces$
             $\production{DistinctConstants}\code{;}$
             $\passume\code{(}\bigwedge_{i \in [m]} \code{w}_i \code{ = t}_{\code{w}_i} \code{);}$
             $\passume\code{(}\bigwedge_{i \in [n]} \code{s}_i \code{ = 0} \vee \code{s}_i \code{ = 1);}$
             $\passume\code{(}\codekey{s}' \code{ = } \codekey{s} \wedge \codekey{s}'' \code{ = } \codekey{s} \code{);}$
             $\production{SetPred}\code{;}\production{SetSucc}\code{;}$
             $\code{\codekey{b}} \passign \code{\codekey{0};}$
             $\pwhile$ #\code{(b$_{n+1} \neq$ 1)}#
                     $\pif \code{(\codekey{b} = \codekey{0})}$
                             $\code{cell} \passign \code{w}_1$
                     $\pelse\,\, \pif \code{(\codekey{b} = \codekey{1})}$
                             $\code{cell} \passign \code{w}_2$
                      $\ldots$
                     $\pelse\,\, \pif \code{(\codekey{b} = \codekey{n})}$
                             $\code{cell} \passign \code{w}_n$
                     $\pelse$
                             $\code{cell} \passign \code{t}_{\code{blank}}$
                     $\pif \code{(\codekey{b} = \codekey{s}}'\code{)}$
                             $\code{c}_1 \passign \code{cell}$
                     $\pif \code{(\codekey{b} = \codekey{s}}\code{)}$
                             $\code{c}_2 \passign \code{cell}$
                     $\pif \code{(\codekey{b} = \codekey{s}}''\code{)}$
                             $\code{c}_3 \passign \code{cell}$
                     $\production{Increment}$
             $\code{turn} \passign \code{1}$

  $\production{DistinctConstants} \produces$
             $\passume$ #\code{(0 $\neq$ 1);}#
             $\passume\code{(}\bigwedge_{i, j \in [m], i \neq j} \code{w}_i \neq \code{w}_j\code{);}$
             $\passume\code{(}\bigwedge_{i, j \in [k], i \neq j} \code{t}_i \neq \code{t}_j\code{)}$
\end{lstlisting}
\end{minipage}  
\vrule \hspace{0.25cm}
\begin{minipage}[t]{.45\textwidth}
\begin{lstlisting}[basicstyle=\scriptsize,mathescape=true,escapechar=\#]
  $\production{Increment} \produces$
             $\pif$ #\code{(b$_1$ = 0)}#
                     #\code{b$_1 \passign$ 1}#
             $\pelse$
                     #\code{b$_1 \passign$ 0}#
                     $\pif$ #\code{(b$_2$ = 0)}#
                             #\code{b$_2 \passign$ 1}#
                     $\pelse$
                             #\code{b$_2 \passign$ 0}#
                             $\pif$ #\code{(b$_3$ = 0)}#
                             $\ldots$

  $\production{SetPred} \produces$
             $\pif$ #\code{(s$_1'$ = 0)}#
                     #\code{s$_1' \passign$ 1}#
                     $\pif$ #\code{(s$_2'$ = 0)}#
                             #\code{s$_2' \passign$ 1}#
                             $\ldots$
                     $\pelse$
                             #\code{s$_2' \passign$ 0}#
             $\pelse$
                     #\code{s$_1' \passign$ 0}#

  $\production{SetSucc} \produces$
             $\pif$ #\code{(s$_1''$ = 1)}#
                     #\code{s$_1'' \passign$ 0}#
                     $\pif$ #\code{(s$_2''$ = 1)}#
                             #\code{s$_2'' \passign$ 0}#
                             $\ldots$
                     $\pelse$
                             #\code{s$_2'' \passign$ 1}#
             $\pelse$
                     #\code{s$_1'' \passign$ 1}#
\end{lstlisting}
\end{minipage}
\vrule
\caption{\figlabel{Grammar2} \emph{Prelude} encodes the input tape
  symbols in the variables \code{w$_i$}, where $\code{t}_{\code{w}_i}$
  denotes the corresponding tape symbol constant. It also establishes
  the secret tracked cell index in variables \code{s$_j$}.  The
  \pwhile-loop prepares the initial tape configuration. To the right
  are the rules for bitvector operations \emph{Increment},
  \emph{SetPred}, and \emph{SetSucc}.}
\end{figure}

\subsection{Correctness}
\label{sec:correctness}

\begin{theorem}
  \thmlabel{lower-bound-correctness} $M$ accepts input $w$ iff there
  exists a coherent program $\, p \in \gram_{M, w} \,$ that is
  correct.
\end{theorem}

\begin{proof}
  \textbf{($\Rightarrow$)} Let $T$ be an accepting computation tree
  for $M$ on $w$. We are to produce a coherent $p \in \gram_{M, w}$
  that satisfies its assertions. We can build $p$ during a pre-order
  traversal of $T$. Note that $p$ will contain branches under the
  \emph{Move} rule that are not taken under any data model, and hence
  do not affect correctness. This is simply by dint of the fact that
  it is only one player's turn in each step. It does not matter how
  these branches are synthesized, so we can take them to all use the
  \emph{Base} rule. We now proceed to describe only those program
  branches that are relevant to ensuring that all assertions hold. The
  \emph{Prelude} rule corresponds to the root of $T$ and produces the
  initial configuration, from which the initial window contents are
  set in \code{\codekey{c}}.

  During traversal, if we reach in $T$ a configuration
  $c_{j-1} \,\,\, (j > 0)$ that is existential, choose
  $\code{dir} \passign \code{0}$ (left) in the $\pif$ branch of
  \emph{Move} if it is the case that $c_j$ is the left child of
  $c_{j-1}$. Otherwise choose $\code{dir} \passign \code{1}$
  (right). In \emph{ProduceCell} under \emph{Generate}, expand such
  that every possible cell index (valuation of \code{\codekey{b}}) is
  branched upon. Choose $\code{cell} \passign \code{t}_s$ in the leaf
  of the branch for \code{(\codekey{b} = \codekey{i})}, where
  $(c_j)_i = t_s$. It is not hard to see that for every interpretation
  of the index \code{\codekey{s}} (where \code{\codekey{s}} ranges
  over possible cells to track) the transition assertion in
  \emph{Check} will indeed hold. We have made sure to select the
  correct choice for \code{dir} and the transitions in the computation
  tree are correct by assumption. If $c_{j-1}$ is a universal
  configuration with left and right children $c_j^l$ and $c_j^r$ in
  $T$, proceed in a similar manner to that of the existential case for
  the left \emph{Generate} subtree in the $\pelse$ branch of
  \emph{Move}, and upon returning to this branch in the traversal, do
  the same for the \emph{Generate} subtree on the right. Upon
  encountering a leaf of $T$ (which is an accepting configuration),
  take the terminating \emph{Base} rule under \emph{Move}. This
  asserts that any cell containing a machine state symbol in fact
  contains the accept state symbol. Since our procedure only ever
  produces cell contents corresponding to valid machine configurations
  that proceed according to the transition relation, it is the case
  that in every interpretation in which the tracked cell contains a
  state symbol we have $\code{c}_2' \code{ = q}_{acc}$.

  Any program from $\gram_{M, w}$ will be coherent, as noted in our
  discussion about Boolean programs in~\secref{further-results}. The
  grammar ensures that no memoizing failures are possible, since every
  variable is effectively Boolean, with the exception of the hardcoded
  machinery for reading universal moves from the data model. In that
  case, terms are computed in a linear fashion and there is no chance
  for recomputation. Finally, all assumes are early by virtue of the
  fact that no variables appearing in equality conditions are ever
  used in a computation with an uninterpreted function.

  \textbf{($\Leftarrow$)} Suppose we have the derivation tree of a
  coherent program $p \in \gram_{M, w}$ such that $p$ satisfies its
  assertions. We are to show there is an accepting computation tree
  for $M$ on $w$. Consider two cases:

  \begin{description}[style=unboxed,leftmargin=0cm]
  \item[Case: no moves.] In this case $p$ does not make any moves,
    which corresponds in its derivation tree to the \emph{Move}
    production immediately rewriting to \emph{Base}. Since $p$ is
    correct, it satisfies its assertions (in all data models) and, in
    particular, it satisfies them in a model where the tracked cell
    contains the initial state symbol $\code{q}_{init}$. That is, the
    success of $\passert\code{(c}_2 \code{ = q}_{acc}\code{)}$ implies
    that $q_{init} = q_{acc}$. Hence $M$ has the initial configuration
    as a trivial accepting computation tree on $w$.

    \vspace{0.1in}
  \item[Case: some moves.] Let us think of the depth of the derivation
    tree for $p$ only in terms of the rules \emph{S}, \emph{Move},
    \emph{Generate}, and \emph{Prelude}. This allows us to speak of
    the number of moves in $p$ in terms of the depth of its derivation
    tree. Now, suppose $p$ has a derivation tree in $\gram_{M, w}$ of
    depth $2m+2$, with $m > 0$ ($m$ is the number of moves). We build
    an accepting computation tree as follows.

    \vspace{0.1in}
    \begin{description}[style=unboxed,leftmargin=0.2cm]
    \item[Base:] The root of the budding computation tree is
      $c_0 = a_1 \ldots a_{2^n}$, where $a_i$ is the tape symbol
      corresponding to the choice for the $i$th cell in the
      \emph{Prelude} loop. We can simulate the loop (which is
      bounded) to obtain the cell contents. It is clear that this is
      the proper initial configuration for $M$ on $w$.  \vspace{0.1in}
    \item[Inductive:] The inductive case proceeds similarly to the
      pre-order traversal from the other direction of our
      proof. Universal turns involve building two branches of the
      computation tree, whereas existential turns only build one. Once
      again, we ignore the infeasible branches of the derivation tree,
      and we can determine which branches these are by keeping track
      of the \code{turn} variable during the traversal. For
      existential turns in the derivation tree, we build the next
      configuration by inspecting the cell choices (via simulation of
      the bounded loop) in the \emph{Generate} subtree (under the
      $\pif$ branch of \emph{Move}). Suppose the configuration
      $c_{next}$, so generated, were not correct. That is, $c_{next}$
      does not follow according to the transition relation $\delta_W$
      from the computation tree parent $c_{prev}$. Then there must be
      some index $k$ for which cell $k$ of $c_{next}$ is wrong
      according to $\delta_W$. But there is then a model where
      \code{(\codekey{s} = \codekey{k})} holds, and hence the
      corresponding assertion inside \emph{Check} fails,
      contradicting the correctness of $p$. Universal turns in the
      derivation tree are processed similarly to existential turns by
      first traversing the left \emph{Generate} subtree and later
      the right.
    \end{description}
  \end{description}

  Finally, since the derivation tree is finite, the last derivation on
  every branch gives the assert for $\code{q}_{acc}$. There is a model
  where the tracked cell contains the final machine state symbol. In
  that model, the satisfaction of the assertion ensures that the
  configuration at every leaf in the computation tree is indeed an
  accepting configuration.
\end{proof}

\newpage

\section{Proof of Lower Bound for Transition System Synthesis}
\seclabel{transition-system-proof-gist}

We show that the realizability and synthesis problems are
$\exptime$-hard using a reduction from the membership problem for
alternating Turing machines with polynomial space bound.  The goal of
the reduction is to design a specification ($\wordaut{\Aa}_R$ and
$\wordaut{\Aa}_S$) such that a correct and coherent transition system
that satisfies it will witness an accepting computation for the Turing
machine.

The key to modeling the desired Turing machine (TM) semantics in
$\wordaut{\Aa}_R$ is to observe that there is a relationship between
the transitions of a specification automaton $\wordaut{\Aa}_R$ and the
nodes of a transition system $TS$ that satisfies it. Notice that the
only way for our transition systems to produce executions containing
$\passume(x=y)$ is to branch at a $\pcheck(x=y)$ node. Thus, any
execution ending with an equality assumption is always accompanied by
a correponding execution ending with a disequality assumption
instead. As in the program synthesis reduction grammar, we want to
restrict our attention to only certain data models. For example, we
want to make sure that the variables we use to model the TM tape cells
initially contain the input symbols. In the program case, we used
statements of the form $\passume(x=y)$ to achieve this. Here however,
we introduce rules in the transition relation for $\wordaut{\Aa}_R$
that allow reading either $\passume(x=y)$ or $\passume(x \neq y)$. The
state reached by reading the negated condition ($x \neq y$ in this
example) will be an accepting state for $\wordaut{\Aa}_R$. This
reflects the fact that we are uninterested in requiring anything of
executions where the TM does not begin with the appropriate input
symbols on its tape. See~\figref{ModelingAssume} for a picture that
illustrates this kind of modeling. Assertions can be modeled in a
similar way. Recall that, besides assignment statements, our
transition systems are restricted to checking equality and disequality
conditions and asserting false. Thus, to model $\passert(x=y)$ a
transition system would first branch on $\pcheck(x=y)$, proceeding
with computation in the affirmative branch and reaching
$\passert(\pfalse)$ in the negative branch. Such assertions can be
enforced in $\wordaut{\Aa}_R$ by introducing transitions for
$\passume(x=y)$ and $\passume(x \neq y)$, with the latter
transitioning to an accepting state after reading
$\passert(\pfalse)$. See~\figref{ModelingAssert} for a picture that
illustrates this kind of modeling.

\begin{wrapfigure}{R}{0.55\textwidth}
\begin{center}
\includegraphics[width=0.54\textwidth,angle=0]{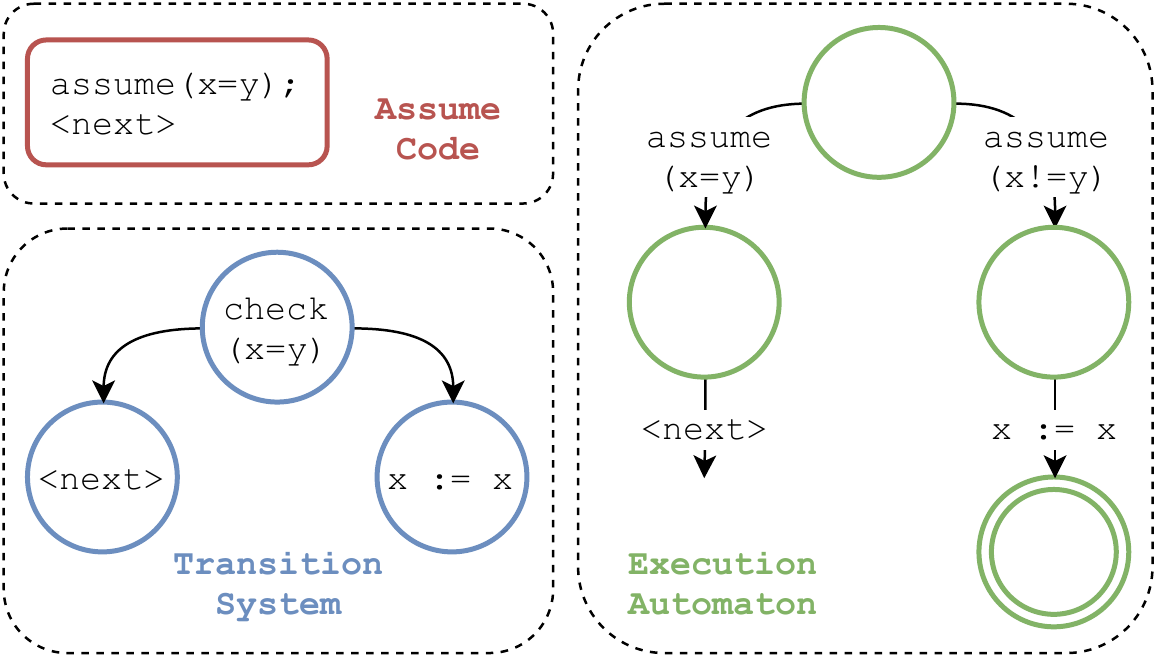}
\caption{ \figlabel{ModelingAssume} To model the
  \textcolor{drawioRed}{Assume Code}, use the
  \textcolor{drawioBlue}{Transition System}, which can be specified
  with the \textcolor{drawioGreen}{Execution Automaton}. Note that the
  assignment statement could be replaced by any arbitrary statement
  that has no bearing on correctness, and \textit{next} represents
  whatever code may follow the assume. }
\vspace{0.2in}
\includegraphics[width=0.54\textwidth,angle=0]{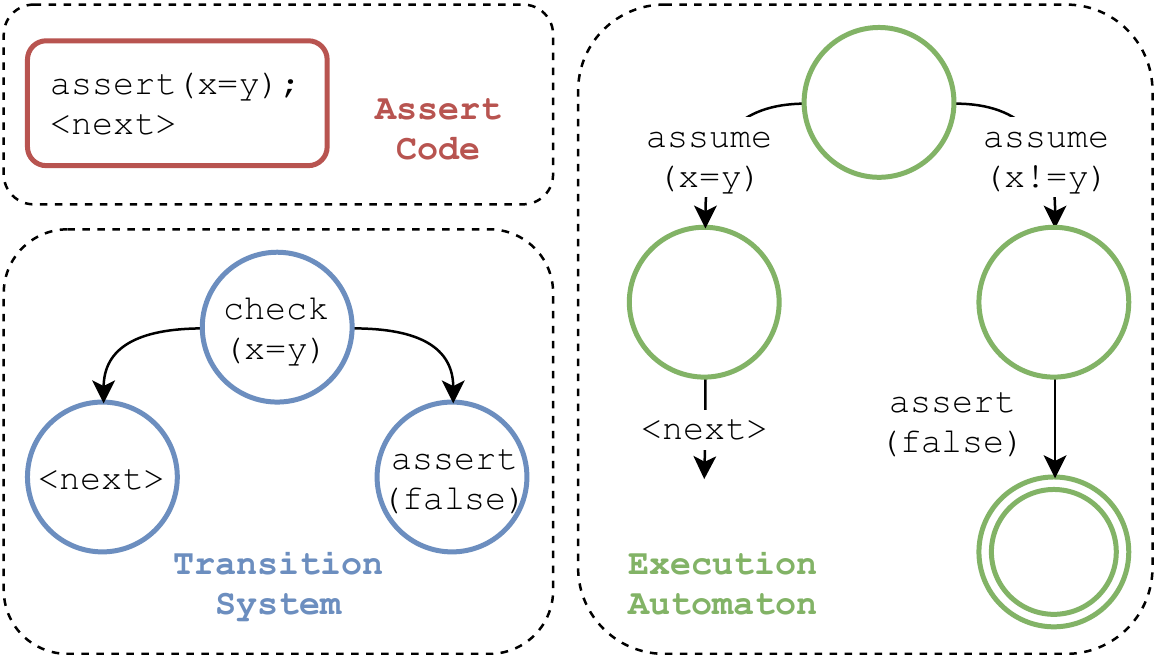}
\caption{ \figlabel{ModelingAssert} To model the
  \textcolor{drawioRed}{Assert Code}, use the
  \textcolor{drawioBlue}{Transition System}, which can be specified
  with the \textcolor{drawioGreen}{Execution Automaton}, where
  \textit{next} represents whatever code may follow the assert. }
\end{center}
\end{wrapfigure}

Having made the observation that some of the components from the
program synthesis reduction grammar can be modeled in the transitions
of the specification automaton $\wordaut{\Aa}_R$, we emphasize once
more the key difference between program and transition system
synthesis. If we attempted to recreate the $\twoexptime$-hardness
proof in this setting, we would be unable to hide information from the
synthesizing algorithm. Imagine that we try using variables to store
the secret index of the tape cell being checked. In order for these
variables to serve the purpose of the lower bound proof, they will
eventually be involved in a $\pcheck$ node. This has the effect of
permanently leaking their values to the synthesis algorithm, which can
make synthesis decisions on the basis of that information. Indeed,
program specification in terms of grammars allows one to enforce the
uniformity of synthesized code, whereas specification in terms of
acceptable executions does not. This leads to an easier problem. 

We now give an overview of the reduction from alternating
polynomial-space TMs. The structure is quite similar to that of the
reduction for grammer-restricted program synthesis, and we hence omit
many details.

\subsection{Gist of the Reduction}
Given an alternating TM $M$ with polynomial space bound, and input $w$
with $|w| = m$, we must construct a specification consisting of
deterministic execution automata $\wordaut{\Aa}_R$ and
$\wordaut{\Aa}_S$ such that there is a correct and coherent transition
system satisfying the specification exactly when $M$ accepts $w$. We
will assume that $M$ uses a counter to ensure its termination in
$2^{\textit{poly(m)}}$ time. Let us now discuss the key aspects of
$\wordaut{\Aa}_R$. We omit the full description, and instead compare
the main components to the corresponding ones in the reduction grammar
for program synthesis. Note that $\wordaut{\Aa}_S$ will accept the
prefix closure of $L(\wordaut{\Aa}_R)$. It can be constructed by a
simple modification that makes final any automaton state that is part
of a path from the initial state to a final state in
$\wordaut{\Aa}_R$.

The states and transitions of $\wordaut{\Aa}_R$ can be divided based
on which of four main goals they serve. The first goal is to
initialize the program variables that model the TM tape cells. This
can be accomplished with a polynomial number of $\passume$
transitions, as mentioned earlier and depicted
in~\figref{ModelingAssume}. The second goal is to simulate the
decision for the next player (either Adam or Eve). Similar to the
program synthesis reduction grammar, we can use variables $turn$ and
$dir$ to model the TM alternation and the transition choice,
respectively. The third goal is to facilitate the generation of new
tape contents after a transition decision is made. Similar to the
program grammar reduction, the tape contents are produced iteratively
and the correctness of each choice is checked by referring to a
sliding window of three previous tape cells. An important difference
is that we can store the entire tape contents in variables, since the
TM uses at most polynomial space. Further, each cell can be updated in
sequence without any counter. There are a polynomial number of tape
cell updates for any given machine transition, and each update can be
specified and checked with a bounded amount of nested branching with
$\passume$ transitions. After each cell is updated, we can model an
$\passert$ to ensure the update follows the appropriate TM
transition. Recall that adequate memory was ensured in the grammar
reduction by requiring programs to branch on transition
choices. Analogously, in $\wordaut{\Aa}_R$ we include transitions for
$\passume(dir = 0)$ and $\passume(dir \neq 0)$ to enforce branching,
and thus adequate memory. Following the polynomially large sequence of
transitions for choosing tape symbols and checking their correctness,
the automaton returns to the $turn$ and $dir$ machinery to (possibly)
repeat the process.

At some point, rather than simulating another turn and transition, the
execution should be allowed to finish. Once again, this is handled in
a similar way as in the program case. We ensure that the contents of
the tape variables indeed constitute an accepting configuration by
checking that if a tape variable contains a state symbol, then that
symbol must be the accepting state symbol. This can easily be
accomplished in a manner similar to checking correctness of tape
transitions. Note that all automaton nodes discussed above involve
implicit transitions to an absorbing reject state for all unmentioned
letters. Each of the four goals outlined above can be easily
implemented with number of states polynomial in the size of the TM
input. Furthermore, the prefix automaton $\wordaut{\Aa}_S$ can be
constructed by making every state accepting, with the exception of the
absorbing reject state.

\subsection{Correctness}
First suppose there is an accepting computation tree $T$ for $M$ on
input $w$, with $|w| = m$. We must show there is a correct transition
system $TS$ whose executions are contained in the language of the
automaton $\wordaut{\Aa}_R$ described above. The transition system
somewhat resembles $T$. It models each alternation in $T$ with a
$\pcheck(turn = 1)$ node, branching to simulate a move from either
Adam or Eve according to the transitions in $T$. To simulate Eve, the
system can choose an assignment node labeled by $dir\,\passign\,0$ or
$dir\,\passign\,1$, depending on the corresponding transition in
$T$. To simulate Adam, the system must go to a sequence of assignment
nodes to read a move decision from the uninterpreted function
$choice$, as in the grammar reduction. After this, the system is
forced (by construction of $\wordaut{\Aa}_R$) to generate correct
updates for each tape cell variable, depending on which transition
decision was made. Each choice can be determined by referring to
configurations in the corresponding branches of $T$. Producing each
tape symbol is accomplished with a sequence of $\pcheck$ and
assignment nodes. Finally, since $T$ contains correct transitions, all
of the complete executions for $TS$ resulting from transition
correctness checks will be correct. Similarly, complete executions
arising from checking that an ending state is accepting will also be
correct, since the leaves of $T$ are accepting configurations.

In the other direction, given a correct transition system $TS$ whose
complete executions are in $L(\wordaut{\Aa}_R)$ and whose partial
executions are in $L(\wordaut{\Aa}_S)$, an accepting computation tree
$T$ for $M$ on $w$ can be built by simulating $TS$ (as in the program
synthesis lower bound proof). Each transition in $T$ will proceed
according to the transition relation for $M$ because $TS$ (which is
correct) has executions that assert this. Since $TS$ has executions
that assert the final tape contents constitute accepting
configurations, every leaf of $T$ will indeed be an accepting
configuration. Note we have assumed that machine $M$ keeps a counter
to ensure termination in exponential time. Thus no correct $TS$ that
satisifes the specification can go on simulating (without halting)
beyond this time bound, since it correctly simulates all machine
transitions. Finally, all executions allowed by $\wordaut{\Aa}_R$ are
easily seen to be coherent.

\newpage

\section{Synthesizing Recursive Programs}
\applabel{recursive-pgm}

We extend the positive result of~\secref{coherent-synth} to synthesize
coherent, recursive uninterpreted programs. The setup for the problem
is very similar. Given a grammar that identifies a class of (now)
\emph{recursive} programs, the goal is to determine if there is a
program in the class that is coherent and correct. We now introduce
the class of recursive programs, their semantics, and some other
background.

\subsection{Recursive Programs and their Semantics}

To keep the presentation simple, we impose restrictions on the program
syntax, none of which limit the generality of our results. Let us fix
the set of program variables to be $V = \set{v_1,v_2,\ldots, v_r}$,
along with an ordering $\tuple{V} = v_1, v_2, \ldots, v_r$. The
programs we consider will have recursively defined methods, and we fix
the names of such methods to belong to a finite set $M$. We will
assume that $m_0 \in M$ denotes the ``main'' method, which is invoked
when the program is executed. Without loss of generality, we assume
the set of local variables for any method $m \in M$ is $V$; methods
can ignore variables if they like. We also assume that the set of
formal parameters for each method is $V$, called in the order
$\fixperm{V}$. None of these are serious restrictions. Methods return
multiple values, which are assigned by the caller to local
variables. Therefore, for every method $m$, we fix $\out{m}$ to be the
(ordered) \emph{output variables} of $m$; the variables in $\out{m}$
will be among the variables in $V$. We require variables in $\out{m}$
to be distinct to avoid implicit aliasing. A recursive program is a
sequence of method definitions, where each definition can make use of
other methods. The formal grammar for recursive programs is given
below.
\[
\begin{array}{rl}
\pgm_{M,V} ::=& m \poutputs \out{m} \stmt_{M,V} \, \mid \,
\pgm_{M,V} \, \pgm_{M,V}\\

\stmt_{M,V} ::=&
\pskip \,
\mid \, x \passign y \,
\mid \, x \passign f(\vec{z}) \, \\
&\mid \, \passume \, \big(\cond_V\big) \, \\
&\mid \, \passert \, \big(\cond_V\big) \\
&
                                           \mid \, \stmt_{M,V} \, ;\, \stmt_{M,V} \\
  &\mid \, \pif \, \big(\cond_V\big) \, \pthen \, \stmt_{M,V} \, \pelse \, \stmt_{M,V} \\
&
\mid \, \pwhile \, \big(\cond_V\big) \, \stmt_{M,V} \,
\mid \, \vec{w} \passign m(\fixperm{V}) \\

\cond_V ::=&
x = y \,
\mid \, x \neq y
\end{array}
\]
Here $x,y,\vec{z}, \vec{w}$ belong to $V$, and the length of vector
$\vec{w}$ must match the output $\out{m}$. A program is simply a
sequence of method definitions, and the main method $m_0$ is invoked
first when the program is run. The most important new statement in the
grammar is $\vec{w} \passign m(\fixperm{V})$, wherein a method $m$ is
called and the return values are assigned to the vector of variables
$\vec{w}$.

Different aspects associated with the semantics of such programs, like
executions, terms, coherence, etc., are given below. Precise
definitions for many of these concepts were first presented
in~\cite{MMV19}. Executions of recursive programs are words over the
alphabet $\Pi_V$, along with two other sets of symbols
\begin{align*}
  \setpred{\dblqt{\pcall~m}}{m \in M} \text{ and }
  \setpred{\dblqt{\vec{z}\,{\passign}\,\preturn}}{\vec{z}\in V^r\,
  \text{ $r\in [\,\mid V\mid\,]$}},
\end{align*}
which correspond to invoking a method and returning from it, and
assigning the outputs to local variables in the caller. The set of
executions of a recursive program is easy to define, and it is a
context-free language.

\subsection{Executions}

The recursive programs in this language have a natural call-by-value
semantics, given a data model that interprets constants and function
symbols. Executions are words over the following alphabet.
\begin{multline*}
\Pi_{M,V} = \{
\dblqt{x \passign y},
\dblqt{x \passign f(\vec{z})},
\dblqt{\passume (x=y)},
\dblqt{\passume (x\neq y)},
\\ \dblqt{\passert (\pfalse)}, \dblqt{\pcall \; m},
\dblqt{\vec{z} \passign \preturn}
\: |\,\,
x, y\in V, \vec{z}\in V^r, r\in [\,\mid V\mid\,], m \in M\}.
\end{multline*}
The (complete) executions of a recursive program $P$ form a
context-free language defined by the following grammar. For each
method $m \in M$, we denote by $s_m$ the body, written over the
grammar $\stmt_{M, V}$, in the definition of $m$. The grammar that
defines the set of executions has non-terminals of the form $X_s$ for
each $s \in \stmt_{M, V}$ appearing in the program text $P$. The
grammar rules are given below.
\begin{align*}
\begin{array}{rcl}
X_{\epsilon} &\goesto& \epsilon\\
X_{\pskip ; s} & \goesto & X_s\\
X_{ x \passign y ; s} &\goesto & \dblqt{x \passign y} \cdot X_s\\
X_{ x \passign f(\vec{z}) ; s} &\goesto& \dblqt{x \passign f(\vec{z})} \cdot  X_s \\
X_{ \passume (c) ; s} &\goesto & \dblqt{\passume(c)}\cdot X_s\\
X_{\passert(\pfalse) ; s} &\goesto & \dblqt{\passert(\pfalse)}\cdot X_s\\
X_{\pif~(c)~\pthen~s_1~\pelse~s_2; ~~s} &\goesto& \dblqt{\passume(c)} \cdot  X_{s_1; s} \\
X_{\pif~(c)~\pthen~s_1~\pelse~s_2; ~~s} &\goesto& \dblqt{\passume(\neg c)} \cdot X_{s_2; s} \\
X_{\pwhile~(c)\{s_1\};~s} &\goesto& \dblqt{\passume(c)} \cdot  X_{s_1; ~\pwhile~(c)\{s_1\};~s} \\
X_{\pwhile~(c)\{s_1\};~s} &\goesto& \dblqt{\passume( \neg c)} \cdot  X_{s} \\
X_{\vec{z} \passign m(\fixperm{V});~s} &\goesto& \dblqt{\pcall~ m} \cdot X_{s_m} \cdot \dblqt{\vec{z} \passign \preturn}
\cdot X_s
\end{array}
\end{align*}
The set of executions of a program $P$, denoted $\exec(P)$, is the
language of the above grammar using start symbol
$X_{s_{m_0};\epsilon}$.  The set of \emph{partial executions}, denoted
$\pexec(P)$, is the set of prefixes of executions in $\exec(P)$.

Observe that all production rules except the one involving method
calls are right linear. 
Further, we can partition the symbols in $\Pi_{M,V}$ into \emph{call},
\emph{return}, and \emph{internal} alphabets:
$\set{\dblqt{\pcall~ m}}$ is the call alphabet,
$\set{\dblqt{\vec{z}\, \passign\, \preturn}}$ is the return alphabet,
and the remaining symbols constitute the internal alphabet. With
respect to such a partition, it is easy to see that $\exec(P)$ is a
\emph{visibly context-free} language~\cite{Alur2009vpa}.

The definition of coherent executions and programs can be extended to
the recursive case. To do this, we need to identify the syntactic term
stored in a variable after a partial execution, the set of syntactic
terms computed during a completed execution, and the collection of
equality assumptions made during an execution. These are natural
extensions of the corresponding concepts in the non-recursive case to
a call-by-value semantics. Based on these notions, coherence is
defined in exactly the same manner as in the non-recursive case.  That
is, executions are coherent if they are memoizing and have early
assumes, and programs are coherent if all their executions are
coherent. We skip the formal definition.

\subsection{Recursive Programs as Trees}

As in the non-recursive case, we represent recursive programs as
finite trees, and the synthesis problem reduces to searching for a
program tree that is coherent, correct, and is generated by a given
grammar $\gram$. Here we describe how we represent recursive programs
as trees. Recall that a recursive program is nothing but a sequence of
method definitions. Therefore, our tree representation of a program is
a binary tree with root labeled by $\dblqt{\proot}$, where the
right-most path in the tree has labels of the form
$\dblqt{m \poutputs \out{m}}$, and the left child of such a node is
the tree representing the body of the method definition of $m$. Since
a method body is nothing but a program statement, we can use a tree
representation similar to the one used for non-recursive programs. We
give an inductive definition for such a tree. Our description presents
trees as terms with the understanding that $\ell(t_1,t_2)$ represents
a tree with root labeled by $\ell$ and left and right sub-trees $t_1$
and $t_2$, respectively. Implicitly, $\ell(t)$ is a tree with label
$\ell$ and \emph{left} sub-tree $t$. With respect to our grammar,
programs don't have a unique parse, since $;$ is associative. In the
description below, we assume some parsing that resolves this
ambiguity. The tree associated with a program $p$ is
$\dblqt{\proot}(\trlbl(p))$, where $\trlbl(p)$ is defined inductively
below.
\begin{align*}
&\trlbl(\epsilon) = \epsilon \\
&\trlbl(m \poutputs \out{m}\, s\ p') = \dblqt{m \poutputs \out{m}}(\trlbl(s),\trlbl(p'))\\
&\trlbl(\pskip) = \dblqt{\pskip} \\
&\trlbl(x{\passign}y) = \dblqt{x{\passign}y}\\
&\trlbl(x{\passign}f(\vec{z})) = \dblqt{x{\passign}f(\vec{z})} \\
&\trlbl(\passume(c)) = \dblqt{\passume(c)}\\
&\trlbl(s_1\, ; \, s_2) = \dblqt{\pseq}(\trlbl(s_1),\trlbl(s_2)) \\
&\trlbl(\pif\, (c)\, \pthen\, s_1\, \pelse\, s_2) = \dblqt{\pite(c)}(\trlbl(s_1),\trlbl(s_2))\\
&\trlbl(\pwhile\, (c)\, s) = \dblqt{\pwhile(c)}(\trlbl(s)) \\
&\trlbl(\vec{z} \passign m(\fixperm{V})) = \dblqt{\vec{z} \passign m(\fixperm{V})}
\end{align*}


Given a data model, every partial execution naturally maps each
program variable to a value in the universe of the data model. We skip
the formal definition. The notions of an execution being feasible in a
data model (i.e., all assume statements must hold when encountered)
and a program being correct (i.e., all executions of the form
$\rho\cdot\passert(\pfalse)$ are infeasible in all data models) can be
extended naturally to recursive programs.



We conclude this section by restating important observations
from~\cite{MMV19} regarding recursive programs.

\begin{theorem}[\cite{MMV19}]
  \thmlabel{rec-coherent-verif} Given a recursive program $P$,
  checking if $P$ is coherent is decidable in $\exptime$.  Further,
  checking correctness of a coherent recursive program is decidable in
  $\exptime$.
\end{theorem}

The proof of \thmref{rec-coherent-verif} relies on the observation
that there are \emph{visibly pushdown automata}\footnote{Visibly
  pushdown automata~\cite{Alur2004vpa}, in our context, have the
  property that they push one symbol to the stack when reading
  $\pcall~m$, pop one symbol on reading $\vec{z} \passign \preturn$,
  and leave the stack unchanged otherwise.}
for the sets of \emph{coherent executions} and \emph{coherent
  executions that are correct}. These automata, denoted
$\wordaut{\Aa}_{\text{rcoh}}$ and $\wordaut{\Aa}_{\text{rcor}}$,
respectively, are both of size $O(2^{\text{poly}(|V|)})$. Since the
set of all program executions is also a visibly context-free language,
decidability follows from taking appropriate automata intersections
and checking for emptiness. By taking cross-products, we can conclude
there is a visibly pushdown automaton $\wrcc$ of size
$O(2^{\text{poly}(|V|)})$ that accepts the set of all recursive
executions which are both coherent and correct; as in the
non-recursive case, we exploit $\wrcc$ for synthesis.

\subsection{Synthesizing Correct, Coherent Recursive Programs}

The approach to synthesizing recursive programs is similar to the
non-recursive case. Once again, given a grammar $\gram$, the set of
trees corresponding to recursive programs generated by $\gram$ is
regular; let $\treeaut{\Aa}_{\gram}$ be the tree automaton accepting
this set of trees. The crux of the proof is to show that there is a
two-way alternating tree automaton $\trcc$ that accepts exactly the
collection of all trees that correspond to recursive programs that are
coherent and correct. The synthesis algorithm then involves checking
if there is a common tree accepted by both $\treeaut{\Aa}_{\gram}$ and
$\trcc$, and if so constructing such a tree. The latter problem is
easily reduced to tree automaton emptiness, so we next describe how to
construct the automaton $\trcc$.

The construction of $\trcc$ is similar to the construction of
$\treeaut{\Aa}_\cohcor$ in the non-recursive case. On an input tree
$t$, $\trcc$ will generate all executions of the program corresponding
to $t$ by walking up and down $t$ and checking if each one of them is
coherent and correct by simulating $\wrcc$. The challenge is to
account for recursive function calls and the fact that $\wrcc$ is a
(visibly) pushdown automaton rather than a simple finite
automaton. Giving a precise formal description of $\trcc$ will be
notationally cumbersome, and obfuscates the ideas behind the
construction. We give only an outline, and leave working out the
precise details to the reader.

Like in the non-recursive case, $\trcc$ will simulate $\wrcc$ as each
execution is generated. Since $\wrcc$ does not change its stack,
except on $\dblqt{\pcall~m}$ and $\dblqt{\vec{z}{\passign}\preturn}$,
we can simulate $\wrcc$ on most symbols by simply keeping track of the
control state of $\wrcc$. The interesting case to consider is that of
method invocation. Suppose $\trcc$ is at a leaf labeled
$\dblqt{\vec{z}{\passign}m(\fixperm{V})}$. Let $q$ be the control
state of $\wrcc$ after the execution thus far. Executing the statement
$\vec{z}{\passign}m(\fixperm{V})$ gives a partial trace of the form
$\dblqt{\pcall~m}\cdot\rho\cdot\dblqt{\vec{z}{\passign}\preturn}$,
where $\rho$ is an execution of method $m$. Suppose $\wrcc$ on symbol
$\dblqt{\pcall~m}$ from state $q$ goes to state $q_1$ and pushes
$\gamma$ on the stack. Notice that no matter what $\rho$ (the
execution of method $m$) is, since $\wrcc$ is visibly pushdown, the
stack at the end of $\rho$ will be the same as that at the
beginning. Therefore, $\trcc$ will (nondeterministically) guess the
control state $q_2$ of $\wrcc$ at the end of method $m$. $\trcc$ will
send two copies. One copy will simulate the rest of the program (after
$\vec{z}{\passign}m(\fixperm{V})$) from the state $q'$, which is the
state of $\wrcc$ after reading $\dblqt{\vec{z}{\passign}\preturn}$
from $q_2$ and popping $\gamma$. The second copy will simulate the
method body of $m$ to confirm that there is an execution of $m$ from
state $q_1$ to $q_2$. To simulate the method body of $m$, $\trcc$ will
walk all the way up to the root, and then walk down, until it finds
the place where the definition of $m$ resides in the tree. $\trcc$
will also need to account for the possibility that the call to $m$
does not terminate; in this case, it will send one copy to simulate
the body of $m$, and if that body ever terminates, $\trcc$ will
reject. 
Given this informal description, one can say that a state of $\trcc$
will be of the form $(p,q_1,q_2)$, where $q_1$ and $q_2$ are a pair of
states of $\wrcc$ with the intuition that $q_1$ is the current state
of $\wrcc$, $q_2$ is the target state to reach at the end of the
method, and $p$ is some finite amount of book-keeping information
needed to perform tasks like finding an appropriate method body to
simulate, whether the method will return, etc. Thus, the size of
$\trcc$ will be $O(2^{\text{poly}(|V|)})$.

\begin{theorem}
  \thmlabel{rec-synthesis} The program synthesis problem for
  uninterpreted, coherent, \emph{recursive} programs is decidable in
  $\twoexptime$; in particular the algorithm is doubly exponential in
  the number of program variables and \complexitygrammar~in the size
  of the input grammar. Furthermore, a tree automaton representing the
  set of \emph{all} correct, coherent, recursive programs conforming
  to the grammar can be constructed in the same time.  Finally, the
  problem is also $\twoexptime$-hard.
\end{theorem}

The $\twoexptime$ lower bound follows from the non-recursive case
(\secref{lower-bound}).

\end{document}